\documentclass[aps,pra,showpacs,twoside,twocolumn,10pt]{revtex4-2}
\usepackage[colorlinks=true, citecolor=red, urlcolor=blue ]{hyperref}
\usepackage{times,epsfig,amssymb,amsfonts,amsmath,bm,subfigure,mathtools,amsthm,braket,soul,enumitem,color,physics,graphics,graphicx}
\usepackage[normalem]{ulem}
\usepackage{comment}
\usepackage{epstopdf}
\newtheorem{theorem}{Theorem}
\newtheorem{corollary}{Corollary}
\newtheorem{definition}{Definition}
\newtheorem{proposition}{Proposition}

\newtheorem{lemma}{Lemma}

\newtheorem{observation}{Observation}

\usepackage{xcolor}

\newcommand{\ayan}[1]{\textcolor{black}{#1}}
\usepackage{optidef}

\newcommand{\adi}[1]{\textcolor{orange}{#1}}

\begin{document}


\title{
Hierarchies of Gaussian multimode entanglement from thermodynamic quantifiers
}

\author{Mrinmoy Samanta\(^1\), Sudipta Mondal\(^1\), Ayan Patra\(^1\), Saptarshi Roy\(^2\), Aditi Sen(De)\(^1\)}

\affiliation{\(^1\)Harish-Chandra Research Institute, A CI of Homi Bhabha National Institute,  Chhatnag Road, Jhunsi, Allahabad - 211019, India\\
\(^2\) NextQuantum Innovation Research Center, Department of Physics and Astronomy, Seoul National University, Seoul 08826, South Korea}

\begin{abstract}

We develop a thermodynamic characterization of multimode entanglement in pure continuous-variable systems by quantifying the gap between globally and locally extractable work (ergotropy). For arbitrary pure multimode Gaussian states, we prove that the $2$-local ergotropic gap is a faithful entanglement monotone across any bipartition and constitutes a functionally independent upper bound to the R\'enyi-2 entanglement entropy. We further introduce the $k$-ergotropic score, the minimum $k$-local ergotropic gap, and show that it faithfully quantifies multimode entanglement across \(k\) partitions. For pure three-mode  Gaussian states, we derive its closed-form relation 
with the geometric measure for genuine multimode entanglement $(k=2)$, and total Gaussian multimode entanglement $(k=3)$.  For systems with more than three modes, the $k$-ergotropic score becomes a functionally independent measure of multimode entanglement to the standard geometric measures. Our results reveal a direct operational hierarchy linking Gaussian multimode entanglement to work extraction under locality constraints, and provide a computable and experimentally accessible thermodynamic framework for characterizing quantum correlations.


\end{abstract}

\maketitle

\section{Introduction}
\label{sec:intro}

The earliest experimental realizations of quantum technologies, particularly in the field of quantum communication, were achieved using photons \cite{Bennett2014, Ekert1991, Bennett1992, Bouwmeester1997, pan1998}, thereby catalyzing the rapid development of quantum optical devices over the past three decades \cite{Gisin2002, Kimble2008, pan2012rmp, Flamini2019, Zhong2020, Pirandola2020}. Quantum optical systems provide access to several controllable degrees of freedom \cite{Allen1992, Kwiat1995, Mair2001, Braunstein2005, OBrien2009, Weedbrook2012rmp}, including polarization, angular momentum, and continuous-variable (CV) degrees of freedom characterized by  canonically conjugate quadratures 
with continuous spectra. Among these, CV quantum systems have emerged as a prominent and foundational paradigm in quantum information science~\cite{BraunsteinPati2003,CerfLeuchsPolzik2007, Weedbrook2012rmp}, 
supporting a wide range of protocols from quantum communication \cite{cvtele, cvdc, Zhang2024, comrmp, coma, comb} to error correction \cite{cvec1, cvec2, cvec3, cvec4, cvec5, cvec6} and measurement-based computation \cite{cvmbqc1, cvmbqc2,cvmbqc4,cvec7, Lee2025}. 
In particular, the continuous spectra of bosonic quadratures enable the deterministic production and manipulation of entanglement, high-efficiency homodyne measurements \cite{Weedbrook2012rmp}, and compatibility with optical and hybrid 
quantum infrastructures \cite{Lee2025}. As a result, CV platforms play an important role in scalable networked quantum information processing, where multimodal correlations across numerous modes serve as the foundation for quantum communication links and distributed computing 
resources. Reliable characterization and quantification of such multimode correlations is therefore a fundamental requirement for improving CV quantum networks. Within the broad class of CV systems, Gaussian states occupy a position of special importance due to their mathematical elegance and experimental accessibility~\cite{Adesso2014}. As CV architectures scale toward networked implementations, multimode Gaussian correlations become indispensable resources, and their precise characterization is 
is not only of practical relevance but also of foundational significance.


On the other hand, the resolution of Maxwell’s demon paradox \cite{landauer1961, bennett1982, zurek1986, toyabe2010, berut2012} established a deep bridge between physics and information theory by revealing a profound connection between information and thermodynamics. This conceptual link rapidly permeated 
many areas of physics, but found particularly strong resonance in quantum theory, especially in the study of quantum correlations \cite{qc1, qc2}. Quantum thermodynamics \cite{horodecki2013limits, brandao2013resource, brandao2015secondlaws} has since emerged as an operational framework for investigating quantum correlations by linking informational quantities to work extraction, heat dissipation, and energetic costs \cite{Huber2015, Sapienza2019}. 
Within this, {\it ergotropy} \cite{Allahverdyan2004} is defined as the maximum amount of work that can be extracted from a quantum state via unitary operations. Building on this concept, {\it the ergotropic gap}~\cite{egap1},
defined as the difference between globally extractable work and the sum of local contributions, attempts to capture the operational advantage enabled by correlations~\cite{Cao2025}. For finite-dimensional systems, the ergotropic gap has been shown to have connections with both bipartite \cite{egap2} and multipartite entanglement \cite{Puliyil22, Cao2025}. Recently \cite{Polo2025}, the ergotropic gap was shown to be a good measure of entanglement for Gaussian pure states of two modes.

\begin{figure}[ht]
\includegraphics[width=1.25\linewidth]{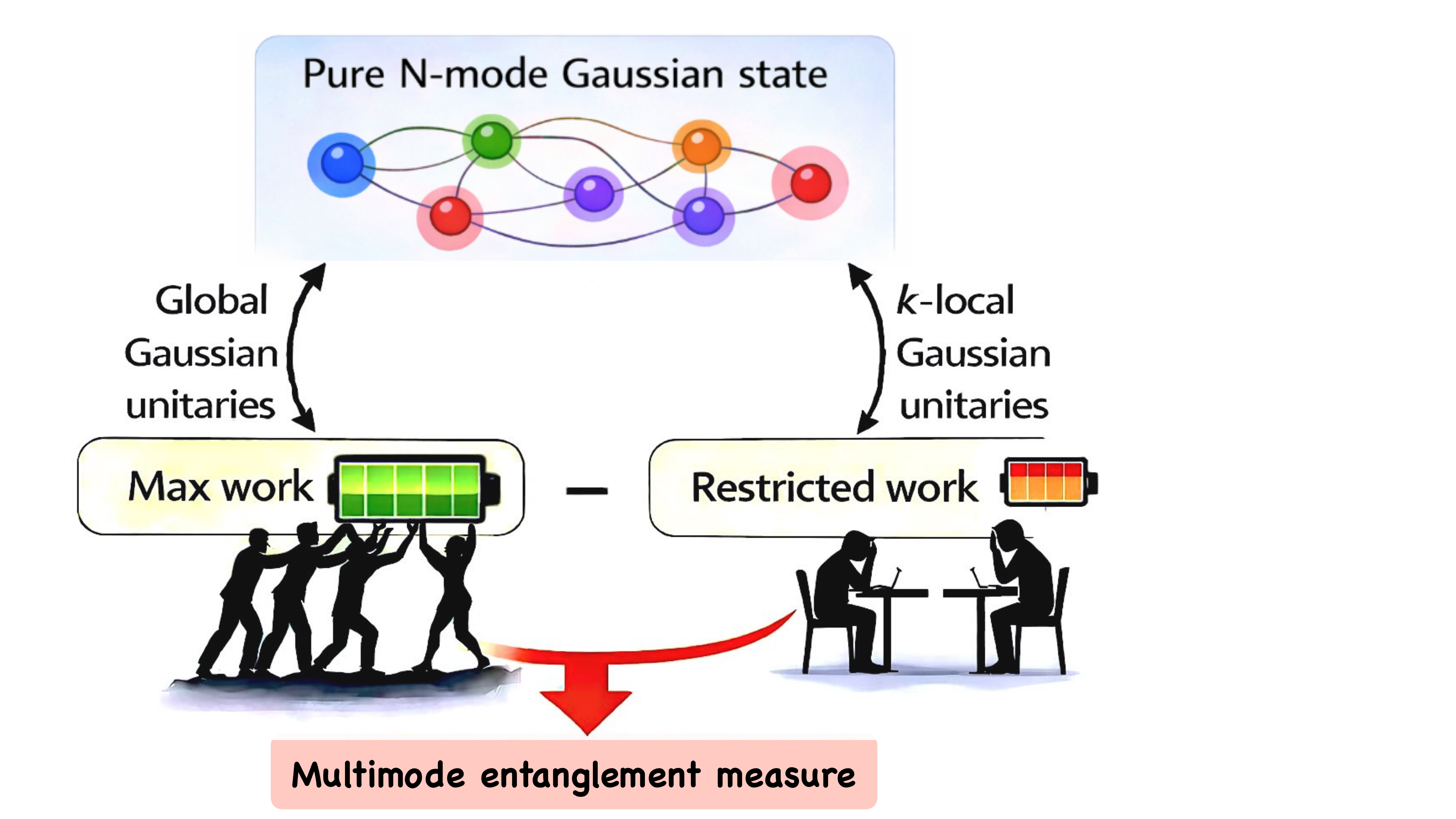}
\caption{
\textbf{Schematic of work extraction from an \(N\)-mode pure Gaussian state.} Global Gaussian unitaries yield the maximal work (global Gaussian ergotropy), while $k$-local Gaussian unitaries yield restricted work ($k$-local ergotropy).
Their difference—the $k$-local ergotropic gap minimized over all $k$ partitions (\(k\)-ergotropic score), quantifies hierarchies of multipartite entanglement of an \(N\)-mode pure Gaussian state.}
\label{fig:schematic} 
\end{figure}

In this paper, we introduce thermodynamic quantities derived from the ergotropic gap and establish them as effective measures for quantifying entanglement of pure multimode Gaussian states in arbitrary partitions (see Fig. \ref{fig:schematic}).  
Specifically, We exhibit that, for arbitrary pure $N$-mode Gaussian states, the $2$-local ergotropic gap provides a faithful quantifier of multimode entanglement across any bipartition of modes. Furthermore, we prove that it constitutes a functionally independent upper bound on the R\'enyi-$2$ entanglement entropy, thereby relating an operationally defined thermodynamic quantity to a standard entropic entanglement.
More importantly, we define the $k$-ergotropic score, which captures a hierarchy of thermodynamic indicators tailored to quantify $k$-partite entanglement in pure multimode Gaussian states. We demonstrate that this score faithfully detects 
entanglement across \(k\) partitions and analyze its behaviour in comparison with conventional distance-based measures of multipartite entanglement, with particular emphasis on both total and genuine multipartite contributions.
In the three-mode case, we obtain an exact closed-form relation between the genuine multimode entanglement and the $2$-ergotropic score, revealing a direct analytical equivalence in this minimal nontrivial setting.
Our results identify work extraction, specifically the difference (gap) between the extracted work (ergotropy) under global and local unitary cycles, as an operational probe of multimode Gaussian entanglement and furnish a computable thermodynamic characterization of quantum correlations.

This paper is organized as follows. In Sec. \ref{sec:framework}, we introduce the definitions of \(k\)-local ergotropic gap and \(k\)-ergotropic score which lead to the main results of this work.  Sec. \ref{sec:2-localegmeasure} establishes \(2\)-local ergotropic gap as a bipartite entanglement measure for pure multimode Gaussian states, valid for bipartitions with an arbitrary number of modes in each subsystem. In Sec. \ref{sec:kloc}, we prove that \(k\)-ergotropic score is an multimode inseparability measure in \(k\) partitions. Moreover, in both these sections, we relate these multimode entanglement measures with known entanglement quantifiers (see Subsecs. \ref{subsec:renyiwith2local} and \ref{subsec:con_ggm}). Finally, we conclude in Sec. \ref{sec:conclusion}.

\section{A framework for \(k\)-ergotropic score}
\label{sec:framework}

We now introduce the mathematical framework for defining ergotropy in multimode
Gaussian systems. 
In this setting, we distinguish different notions of ergotropy according to the class of allowed
operations. 
\begin{definition}[Global Gaussian ergotropy (GGE)]
   \label{def:g_ergo}
The Global Gaussian Ergotropy (GGE) of \(N\)-mode state, $\rho_{_{N}}$ is given by
 \begin{eqnarray}
     \mathcal{E}_G (\rho_{_{N}}) &=& \Tr(\rho_{_{N}}H) - \min_{U} \Tr(U\rho_{_{N}}U^\dagger H),
    \label{eq:global_eg}
 \end{eqnarray}
 where $U$ is a Gaussian unitary acting globally on all the \(N\)-modes.  
\end{definition}
\noindent Performing minimization over \(U\), the term becomes \(\underset{U}{\min} \Tr(U\rho_{_{N}}U^\dagger H)=\Tr(\rho^{^\text{g,p}}H)\), where \(\Tr(\rho^{\text{g,p}}H)=\frac{1}{4}\Tr(\Sigma^{^{\text{g,p}}}-2N)\) with the superscript ``g,p" representing the global passive state. Therefore, Eq.~(\ref{eq:global_eg}) reduces to 
  \begin{eqnarray}
     \mathcal{E}_G (\rho_{_{N}}) 
    &=&\frac{1}{4}(\Tr\Sigma - \Tr\Sigma^{^\text{g,p}}),
    \label{eq:global_eg_f}
 \end{eqnarray}
 where via the Williamson decomposition, \(\Sigma^{^\text{g,p}}=\mathrm{diag}( \mu_1,\mu_1,\mu_2,\mu_2,\ldots,\mu_N,\mu_N)\) and \(\{\mu_i\}_{i=1}^{^N}\) denotes the simplectic eigenvalues of \(\Sigma^{^\text{g,p}}\).

Let us split the \(N\)-mode Gaussian state, $\rho_{_N}$, among \(k\)-partitions (\(1\le k\le N-1\)) specified by $\{\mathcal{A}_1,\mathcal{A}_2,\ldots,\mathcal{A}_k\}$. 
Each partition $\mathcal{A}_i$ is charaterized by the collection of modes it contains. Furthermore, each partition must have atleast one-mode, and no mode can be assigned to multiple partitions. We formally list these properties as
\begin{eqnarray}
  (i)~  |\cup_{i=1}^k \mathcal{A}_i| &=&\overset{k}{\underset{i=1}{\sum}}m_i=N ,\nonumber\\
  \text{and}~(ii)~~~ \mathcal A_i \cap \mathcal{A}_j &=& \phi, ~\forall i \neq j \in[1,k],
  \label{eq:properties}
\end{eqnarray}
where $|\mathcal{A}_i|=m_i$ denotes the number of modes in the partition $\mathcal{A}_i$, and \(\phi\) denotes null set. 
Incidentally, the number of $k$-partitions of an $N$-mode state is simply given by the Sterling's number \cite{abramowitz1972stirling} of the second kind $S(N,k) = \frac{1}{k!}\sum_{j=0}^k (-1)^j \binom{k}{j} (k-j)^N.$


\begin{definition}[ $k$-local Gaussian ergotropy]
The $k$-local Gaussian ergotropy of \(N\)-modes state, $\rho_{_{N}}$ is given by
 \begin{eqnarray}
     \mathcal{E}_{k-{\tt loc}}^{\mathcal{A}_1|\mathcal{A}_2|\cdots|\mathcal{A}_k} (\rho_{_{N}}) &=& 
     \Tr(\rho_{_{N}}H)-\min_{\mathcal U_k =\otimes_{i=1}^k U_{\mathcal{A}_i}} \Tr(\mathcal U_{k}\rho_{_{N}}\mathcal U_{k}^\dagger H),\nonumber\\
     \label{eq:k_loc_ergo}
 \end{eqnarray}
 where \(U_{\mathcal{A}_i}\) is a unitary acting on the modes in the \(\mathcal{A}_i\)-th partition. 
\end{definition}
 Similar to GGE, after performing minimization over \(\mathcal{U}_k\), Eq.~(\ref{eq:k_loc_ergo}) reduces to
\begin{eqnarray}
    \mathcal{E}_{k-{\tt loc}}^{\mathcal{A}_1|\mathcal{A}_2|\cdots|\mathcal{A}_k} (\rho_{_{N}})= \frac{1}{4}(\Tr \Sigma -\Tr\Sigma^{^{k-{\tt loc},p}}),
    \label{eq:k_loc_erg_f}
\end{eqnarray}
where the Williamson decomposion, \(\Sigma^{^{k\text{-}{\tt loc},p}}\) takes the form as 
\begin{eqnarray}
\Sigma^{^{k\text{-}{\tt loc},p}} &=
\begin{pmatrix}
\Sigma_{\mathcal{A}_1} & \Gamma_{\mathcal{A}_1\mathcal{A}_2} & \cdots & \Gamma_{\mathcal{A}_1\mathcal{A}_k} \\
\Gamma_{\mathcal{A}_1\mathcal{A}_2}^{\mathsf T} & \Sigma_{\mathcal{A}_2} & \cdots & \Gamma_{\mathcal{A}_2\mathcal{A}_k} \\
\vdots & \vdots & \ddots & \vdots \\
\Gamma_{\mathcal{A}_1\mathcal{A}_k}^{\mathsf T} & \Gamma_{\mathcal{A}_2\mathcal{A}_k}^{\mathsf T} & \cdots & \Sigma_{\mathcal{A}_k}
\end{pmatrix},
\label{eq:gen_cova}
\end{eqnarray}
with \(\Sigma_{\mathcal{A}_i}=\mathrm{diag}(
\nu_1^{\mathcal{A}_i},\nu_1^{\mathcal{A}_i},\nu_2^{\mathcal{A}_i},\nu_2^{\mathcal{A}_i},\ldots,\nu_{m_i}^{\mathcal{A}_i},\nu_{m_i}^{\mathcal{A}_i})\) and \(\Gamma_{{\mathcal{A}_i}\mathcal{A}_j}\) a real matrix and $\Sigma^{^{k\text{-}{\tt loc},p}}$ denotes the $k$-local passive state.  We are now ready to introduce the central quantities of interest for our work.

\begin{definition}{\emph{($k$-Local ergotropic gap.)}}
   The \(k\)-local ergotropic gap of an $N$-mode Gaussian state $\rho_{_N}$ for a given $k$-partitions $\{\mathcal{A}_1,\mathcal{A}_2,\ldots,\mathcal{A}_k\}$ is defined as the difference of the global ergotropy and \(k\)-local ergotropy, i.e.,
    \begin{eqnarray}
    \delta_{k-\tt{loc}}^{\mathcal{A}_1|\mathcal{A}_2|\cdots|\mathcal{A}_k}(\rho_{_N})&=&\mathcal{E}_G (\rho_{_{N}})-\mathcal{E}^{\mathcal{A}_1|\mathcal{A}_2|\cdots|\mathcal{A}_k}_{k-\tt{loc}} (\rho_{_{N}}),\nonumber\\&=&\frac{1}{4}(\Tr \Sigma^{^{k{-{\tt loc},p}}}-\Tr\Sigma^{^\text{g,p}}), \nonumber \\
   &=& \frac{1}{2}
\big(\sum_{j=1}^{k} \sum_{i=1}^{m_j} \nu_i^{\mathcal{A}_j} - \sum_{k=1}^{N} \mu_k\big).
\label{eq:general_erg_gap_fixed}
\end{eqnarray}
\end{definition}



\ayan{Given that an $N$-mode Gaussian state admits many possible $k$-partitions, we define the $k$-ergotropic score based on the above definition.}

\begin{definition}{\emph{($k$-ergotropic score.)}}
    The $k$-ergotropic score is defined as the minimum value of the $k$-local ergotropic gap,
\begin{eqnarray}
\Delta^{k}(\rho_{_N})&=&\min_{\mathcal{A}_1|\mathcal{A}_2|\ldots|\mathcal{A}_k}\delta ^{\mathcal{A}_1|\mathcal{A}_2|\ldots|\mathcal{A}_k}_{k-\tt loc}(\rho_{_N}),
\label{eq:min_ergo}
\end{eqnarray}
where the minimization is taken over all possible $k$-partitions of $\rho_{_N}$ satisfying the properties in Eq. \eqref{eq:properties}.
\end{definition}

A multimode quantum state can exhibit entanglement in qualitatively different ways. For instance, one may analyze the bipartite entanglement of a multimode state with respect to a chosen bipartition, which, in the case of pure Gaussian states, is fully characterized by the von Neumann entropy of the corresponding reduced states. More generally, a multimode pure state is called $k$-separable if it is separable across at least $(k-1)$ bipartitions, whereas it is said to exhibit $k$-inseparable multipartite entanglement if it fails to be $k$-separable. Notably, $k$-inseparability implies $(k+1)$-inseparability (see Appendix~\ref{subsec:entmeasure} for a detailed discussion of the multipartite entanglement structure). In the following, we aim to quantitatively characterize the $k$-inseparability of multimode pure Gaussian states from a thermodynamic perspective, demonstrating that the $k$-ergotropic score serves as a quantitative measure of the $k$-inseparable multimode entanglement content of such states. 


\section{Measuring bipartite entanglement with \(2\)-local ergotropic gap}
\label{sec:2-localegmeasure}

We are now ready to focus on the $2$-local ergotropic gap of an $N$-mode pure Gaussian state partitioned into two parts (subsystems), $\mathcal{A}_1$ and $\mathcal{A}_2$, containing $m_1$ and $m_2$ modes, respectively, with
$m_1+m_2=N$, and demonstrate that it is a valid measure of  entanglement of a \(N\)-mode pure Gaussian state across a given bipartition.

\begin{lemma}
   For a pure $(m_1+m_2)$-mode Gaussian state $\rho_{_N}$, bipartitioned into subsystems
$\mathcal{A}_1$ (with $m_1$ modes) and $\mathcal{A}_2$ (with $m_2$ modes), having $m_1 \le m_2$, the
$2$-local ergotropic gap is entirely determined by the symplectic spectrum of the
reduced state on the smaller subsystem $\mathcal{A}_1$ as
 \begin{eqnarray}
\delta^{\mathcal{A}_1|\mathcal{A}_2}_{2\text{-loc}}(\rho_{_N})
&=&
\sum_{i=1}^{m_1}\big(\nu_i^{\mathcal{A}_1}-1\big),
\label{eq:gap_pure}
\end{eqnarray}
where $\{\nu_i^{\mathcal{A}_1}\}_{i=1}^{m_1}$ denotes the set of symplectic eigenvalues of
the reduced covariance matrix of partition $\mathcal{A}_1$.
\label{le:lemma_1}
\end{lemma}

\ayan{An outline of the proof of Lemma~\ref{le:lemma_1} is as follows. Since \(\rho_{_N}\) is a pure state, $\mu_k=1$ for all $k$ in Eq.~(\ref{eq:general_erg_gap_fixed}). Moreover, among $m_2$ symplectic eigenvalues of the subsystem $\mathcal{A}_2$, $m_1$ coincides with those of the subsystem $\mathcal{A}_1$, while the remaining $(m_2-m_1)$ are equal to $1$. Details of the proof are discussed in Appendix.~\ref{app:coro_1}. This allows us to obtain one of the central results of our work.}

\begin{theorem}
For a pure \((m_1+m_2)\) modes Gaussian state, the \(2\)-local ergotropic gap, \(\delta^{\mathcal{A}_1|\mathcal{A}_2}_{2\text{-loc}}(\rho_{_N})\), constitutes an entanglement measure in the bipartition \(\mathcal{A}_1\) : \(\mathcal{A}_2\), \ayan{i.e., it vanishes for any product state across
the bipartition $\mathcal{A}_1:\mathcal{A}_2$ and is monotonic under Gaussian local operations and classical communication (GLOCC).}
\label{th:th_1}
\end{theorem}

{\color{black}
\begin{proof}
The vanishing of $\delta^{\mathcal{A}_1|\mathcal{A}_2}_{2\text{-loc}}(\rho_{_N})$ (when $\rho_{_N} = \ket{\psi_{\mathcal{A}_1}}\otimes \ket{\psi_{\mathcal{A}_2}}$) can be directly concluded from the fact that all the symplectic eigenvalues of $\ket{\psi_{\mathcal{A}_1}}$ are unity. Further, Eq. \eqref{eq:gap_pure} renders the $2$-local ergotropic gap to be zero. The proof of $\delta^{\mathcal{A}_1|\mathcal{A}_2}_{2\text{-loc}}(\rho_{_N})$ being non-increasing under GLOCC operations is established by decomposing \cite{Botero03} the pure Gaussian state across the $\mathcal{A}_1|\mathcal{A}_2$ bipartition as a product state involving entangled
two-mode squeezed states and single-mode local states.  Then by applying the weak majorization condition for pure bipartite Gaussian states \cite{Giedke03}, we show that under GLOCC transformations, the sum of symplectic eigenvalues of the modes in partition $\mathcal{A}_1$ remains non-increasing. For the details of the proof, see Appendix.~\ref{app:theory_1_main}.
\end{proof}
}

On top of satisfying the minimal requirements (see Def. \ref{def:requirements} in Appendix \ref{subsec:entmeasure}) for being an entanglement measure, the \(2\)-local ergotropic gap additionally satisfies other ``good" to have properties.
\begin{corollary}
The \(2\)-local ergotropic gap, \(\delta^{\mathcal{A}_1|\mathcal{A}_2}_{2\text{-loc}}(\rho_{_N})\) satisfies the following properties:   1. faithfulness, 2. strong monotonicity, and 3. additivity.
    \label{co:corollary2}
\end{corollary}
The proof of Corollary.~\ref{co:corollary2} is discussed in Appendix.~\ref{app:theorem_1_additional}.

\subsection{Connection with standard measures of entanglement: R{\'e}nyi-$2$ vs \(2\)-local ergotropic gap}
\label{subsec:renyiwith2local}

Let us  now 
connect \(2\)-local ergotropic gap with other information theoretic measures of entanglement, namely 
the Rényi-$2$ entanglement entropy, which in the Gaussian case, satisfies \emph{strong subadditivity}, thereby having a spacial status \cite{strong}. It means, in the Gaussian setting, the central structures of quantum information theory admit a consistent reformulation in terms of the Rényi-$2$ entropy \cite{strong, Adesso2014}. 

\begin{theorem}[$2$-local ergotropic gap bounds Rényi-$2$ entanglement]
For a pure $(m_1+m_2)$-mode Gaussian state bipartitioned as 
$\mathcal{A}_1|\mathcal{A}_2$, with $m_1 \le m_2$, the $2$-local ergotropic gap provides, in general, a functionally independent upper bound to the Rényi-$2$ entanglement entropy,
\begin{equation}
S_2(\mathcal{A}_1{:}\mathcal{A}_2)
\;\le\;
\delta_{2\text{-}{\tt loc}}^{\mathcal{A}_1|\mathcal{A}_2}(\rho_{_N}),
\end{equation}
with equality if and only if the state is separable (product) across 
$\mathcal{A}_1:\mathcal{A}_2$ bipartition.
\label{th:renyi-2upperbound}
\end{theorem}
{\color{black}
\begin{proof}
The upper bound can be established by noticing that for pure states, the R{\'e}nyi-$2$ entanglement entropy equals the R{\'e}nyi-$2$ entropy
of the reduced state, $S_2(\mathcal{A}_1{:}\mathcal{A}_2)
= S_2(\rho_{\mathcal{A}_1})
= \frac{1}{2}\log \det \sigma_{\mathcal{A}_1}
= \sum_{j=1}^{m_1} \log \nu_j.$ 
Using $\log x \leq x-1$, and  Eq. \eqref{eq:gap_pure}, we arrive at $\delta_{2\text{-}{\tt loc}}^{\mathcal{A}_1|\mathcal{A}_2}(\rho_{_N})$.
The functional independence is proven by contradiction. Here we provide a concrete example of a family of states with the same R{\'e}nyi-$2$ entanglement but different $2$-local ergotropic gaps. This rules out the possibility of any functional dependence. A detailed proof is provided in Appendix. \ref{app:th2}. 
\end{proof}
Let us emphasize here the importance for demonstrating the functional independence. It establishes the $2$-local ergotropic gap as an independent measure of bipartite entanglement. Consequently, the proofs of the monotonicity conditions in Def.~\ref{def:requirements} cannot be mapped to those of standard entanglement measures and must therefore be established independently.}
 Similar relations can also be obtained if one considers the von-Neumann entropy as the measure of information, although the bounds obtained are a little less neat. For a detailed analysis, please look at Appendix \ref{app:ee}.

\section{Measuring multimode entanglement with \(k\)-local ergotropic gap}
\label{sec:kloc}

 Instead of dividing \(N\) modes into two,
 let us consider the $N$ modes of $\rho_{_N}$ to be partitioned into $k$ groups: $\{\mathcal{A}_1,\mathcal{A}_2,\ldots,\mathcal{A}_k\}$. For the
case of pure $N$-mode Gaussian states, 
the $k$-local ergotropic gap admits a simplified characterization, which we present in the following lemma.

\begin{lemma}
For a pure $N$-mode Gaussian state $\rho_{_N}$, partitioned into
$\mathcal{A}_1|\mathcal{A}_2|\ldots|\mathcal{A}_k$  with \(\sum_{i=1}^km_i=N,~\text{and}~m_i\geq1\), the $k$-local ergotropic gap
is given by
\begin{eqnarray}
\delta^{\mathcal{A}_1|\mathcal{A}_2|\cdots|\mathcal{A}_k}_{k-\tt loc}
(\rho_{_N}) &=& \frac{1}{2}\sum_{j=1}^{k}\sum_{i=1}^{m_j}\big(\nu_i^{\mathcal{A}_j} - 1\big),
\label{eq:Lema2}
\end{eqnarray}
where $\{\nu_i^{\mathcal{A}_j}\}_{i=1}^{m_j}$ denotes the set of symplectic eigenvalues of
the reduced covariance matrix of the modes in partition $\mathcal{A}_j$.
    \label{co:corollary3}
\end{lemma}
Since \(\rho_{_N}\) is pure  with all \(\mu_k = 1\), it reduces Eq.~(\ref{eq:general_erg_gap_fixed}) to Eq.~(\ref{eq:Lema2}). { Furthermore, using Eq.~(\ref{eq:gap_pure}), the \(k\)-local ergotropic gap, \(\delta^{\mathcal{A}_1|\mathcal{A}_2|\cdots|\mathcal{A}_k}_{k\text{-loc}}(\rho_{_N})\),  can be expressed as the sum of the \(2\)-local ergotropic gaps 
\begin{eqnarray}
\delta^{\mathcal{A}_1|\mathcal{A}_2|\cdots|\mathcal{A}_k}_{k-\tt loc}
(\rho_{_N})&=&\frac{1}{2}\sum_{j=1}^{k}\delta ^{\mathcal{A}_j|\bar{\mathcal{A}_j}}_{_{2-\tt loc}}(\rho_{_N}),
    \label{eq:k_loc_2_loc1}
\end{eqnarray}
where $\bar {\mathcal{A}_j} = \bigcup_{i\neq j}^k \mathcal{A}_i$ contains all modes except those in $\mathcal{A}_j$, and
\(
\delta ^{\mathcal{A}_j|\bar{\mathcal{A}_j}}_{_{2-\tt loc}}(\rho_{_N})=\sum_{i=1}^{m_j}\big(\nu_i^{\mathcal{A}_j} - 1\big),
\)

 The elegant simplification of the $k$-local ergotropic gap offered by Lemma \ref{co:corollary3} enables us to establish our central result about the $k$-ergotropic score.
\begin{theorem}
For a pure \(N\)-mode Gaussian state, the \(k\)-ergotropic score, \(\Delta^{k}(\rho_{_N})\) is a valid measure of $k$-partite entanglement.
    \label{th:th4}
\end{theorem}
{\color{black}
\begin{proof}
The fact that \(\Delta^{k}(\rho_{_N})\) vanishes for any $(k+1)$-product state can be obtained by combining Eq. \eqref{eq:Lema2} with the requirement of the symplectic spectrum being unity for any pure state. 
  Moreover, the \(k\)-local ergotropic gap is  GLOCC monotone since  the \(2\)-local ergotropic gap is monotonic under GLOCC and  \(\delta^{\mathcal{A}_1|\mathcal{A}_2|\cdots|\mathcal{A}_k}_{k-\tt loc}
(\rho_{_N})\) is the sum of all \(k\)-possible \(2\)-local ergotropic gaps as shown in Eq.~(\ref{eq:k_loc_2_loc1}). Therefore,
\(\Delta^{(k)}(\rho_{_N})\) is also a monotone under GLOCC as it is the minimum of all the \(k\)-local ergotropic gap values. See Appendix.~\ref{app:min_ergo_k_party} for a detailed proof.
\end{proof}
}
\noindent 
Interestingly, we find that the quantity \(\Delta^{k}(\rho_{_N})\) additionally satisfies some desirable properties.
\begin{corollary}
For a pure \(N\)-mode Gaussian state, the \(k\)-ergotropic score, \(\Delta^{k}(\rho_{_N})\) is a faithful and additive measure of $k$-partite entanglement.
    \label{co:corollary4}
\end{corollary}
\noindent The proof of Corollary.~\ref{co:corollary4} is discussed in Appendix.~\ref{app:prove_corollary_4}.



\subsection{Connecting \(k\)-mode inseparability with geometric measures of entanglement}
\label{subsec:con_ggm}

Let us establish a connection between the \(k\)-ergotropic score, \(\Delta^k(\rho_{_N})\), and the standard distance-based multimode geometric measure of entanglement, \(\mathcal{E}_k(\rho_{_N})\) \eqref{eq:k_distance_measure}.
Specifically, we will now focus on two particular cases: $1. ~k = 2$, corresponding to genuine multimode entanglement, ~$2. ~k = N$, accounting for the total multimode entanglement.
Since we concentrate only pure states, we can take the distance measure $D$ in \(\mathcal{E}_k(\rho_{_N})\) to be the fidelity distance, where $D(\ket\psi,\ket\phi)=1-|\langle\psi|\phi\rangle|^2$. With this distance measure, the genuine multimode entanglement is referred to as the generalized geometric measure (GGM) \cite{ggm1,ggm2}.

\subsubsection{Relation between GGM and \(2\)-ergotropic score
}
\label{sec:gme_ggm}
The GGM content $(G:=\mathcal{E}_{2})$ of an \(N\)-mode pure Gaussian state \cite{Saptarshi20,Saptarshi25}, \(\rho_{_N}\) can be computed succinctly via the symplectic spectrum of its marginal states as 
\begin{eqnarray}
G(\rho_{_N})
= 1 - \max_{m=1,2,\ldots,[ \frac{N}2 ]}
\;\max_{j=1,2,\ldots,\binom{N}{m}}
\left(
\prod_{i=1}^{m} \frac{2}{1+\nu_i^{\,j}}
\right),
\label{eq:formula_ggm}
\end{eqnarray}
where $\{\nu_i^j\}_{_{i=1}}^{^m}$ are the symplectic eigenvalues of the  $j$-th  state out of the $\binom{N}{m}$ reduced \(m\)-mode states in \(\mathcal{A}_1\) partition. Interestingly, 
a direct one to one relation of \(2\)-ergotropic score, \(\Delta^2(\rho_{_3})\) and GGM, \(G(\rho_{_3})\) for three mode pure states is presented in 
the following proposition:
\begin{proposition}
For any three-mode pure Gaussian state $\rho_3$, the \(2\)-ergotropic score \(\Delta^2(\rho_{_3})\) is a monotonic function of the GGM content of the state \(G(\rho_{_3})\), i.e., 
 \begin{eqnarray}
   \Delta^2(\rho_{_3})&=& \frac{2~G(\rho_{_3})}{1 - G(\rho_{_3})}.
   \label{eq:conn_ggm}
\end{eqnarray} 
\label{pro:relation_ggm_gap}
\end{proposition}
\noindent 
This  one-to-one correspondence arises because, among all bipartitions, the minimum symplectic eigenvalue originates from the same bipartition in both cases (see 
Appendix~\ref{App:relation_ggm_gap} for the detailed proof). 
Interestingly, this equivalence highlights an important point.  Due to their monotonic functional dependence, just like $G(\rho_{3})$, $\Delta^2(\rho_{3})$ is monotonic under local operations and classical communication (LOCC) as opposed to only the GLOCC $\subset$ LOCC operations. However, this special feature holds only for the special case of $N = 3$. In general, for $N \geq 4$, $G(\rho_{N})$ and $\Delta^2(\rho_{N})$ are functionally independent, which can be intuitively seen from their forms in Eqs. \eqref{eq:min_ergo} and \eqref{eq:formula_ggm}. The inequivalence can also be directly observed from our numerical simulations of random four-mode pure Gaussian states in Fig. \ref{fig:relation}(a), where for $N =4$, for a given value of $G$, there is a spread of values for $\Delta^2$, and vice-versa, ruling out the possibility of a functional dependence.


\begin{figure}[ht]
\includegraphics[width=\linewidth]{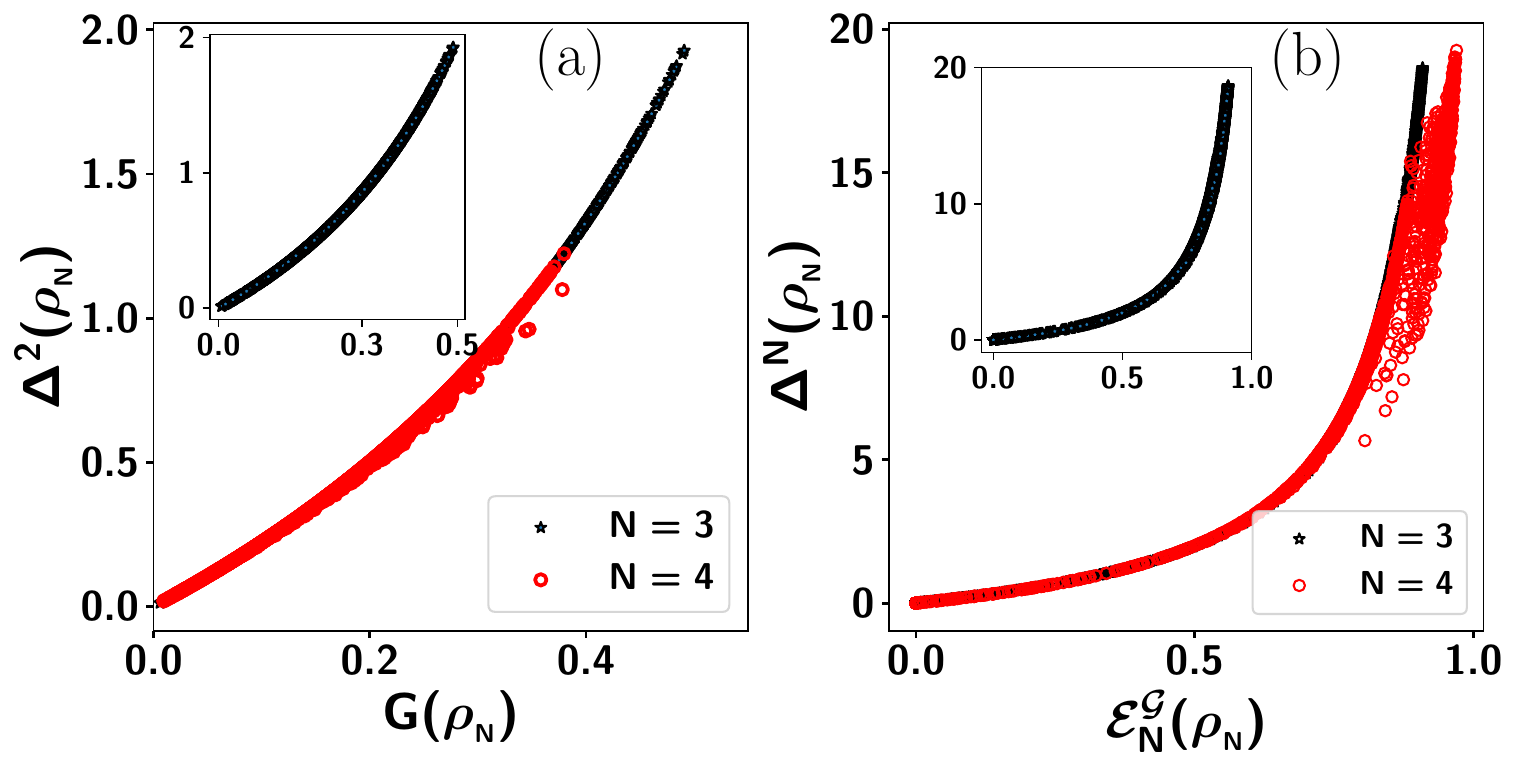}
\caption{
\textbf{Relation between thermodynamic and geometric measures of multimode entanglement.}
We compare multimode entanglement quantified by the $k$-ergotropic score, $\Delta^{k}(\rho_{N})$, (vertical axis) with the generalized geometric measure (GGM), $G(\rho_{N})$, and total multimode entanglement, $\mathcal{E}^{\mathcal{G}}_{N}(\rho_{N})$ (horizontal axis) at fixed total energy, $E=20$. 
(a)  $\Delta^{2}(\rho_{N})$ with respect to $G(\rho_{N})$, for $N=3$ and $N=4$. For $N=3$, these measures exhibit a one-to-one correspondence as shown in Proposition 1 (see inset). 
(b) The total multimode entanglement (TME), $\mathcal{E}^{\mathcal{G}}_{N}(\rho_{N})$ against $N$-ergotropic score,  $\Delta^{N}(\rho_{N})$ for $N=3,~4$. 
In  this case also, we numerically find a one-to-one relation between TME and $N$-ergotropic score for the three-mode scenario (see inset). 
}
\label{fig:relation} 
\end{figure}

\subsubsection{Linking \(N\)-ergotropic score to the total Gaussian multimode entanglement}
\label{sec:total_ent}


We now move to the case \(k = N\), which corresponds to the total multimode entanglement. In particular, we consider a variant of the distance-based measure denoted by $\mathcal{E}_N^{\mathcal{G}}$ which computes the minimum distance of the given state from the set of \(N\)-separable Gaussian states. We denote $\mathcal{E}_N^{\mathcal{G}}$ as the total Gaussian multipartite entanglement (GTME) and can be expressed as follows.
\begin{proposition}
    The total Gaussian multimode entanglement of a pure Gaussian state with vanishing displacement and covariance matrix $\Sigma_{\rho}$ is given by
    \begin{eqnarray}
        \mathcal{E}^{\mathcal{G}}_N(\rho_{_N})= 1 - \max_{\{\vec{r},\vec{\theta}\}}  \Big[\prod_{i=1}^N \nu_i^{-1}(\{\vec{r},\vec{\theta}\})\Big],
\label{eq:distance_easy_form}
    \end{eqnarray}
    where $\nu_i(\{\vec{r},\vec{\theta}\})$ are the symplectic eigenvalues of a Gaussian state with covariance matrix $W(\{\vec{r},\vec{\theta}\})+\Sigma_{\rho}$, where  
     $W(\{\vec{r},\vec{\theta}\})$ is the covariance matrix of $N$-single mode squeezed vacuum states with squeezing parameters $\{r_j e^{i\theta_j}\}_{j=1}^N$. 
\label{pro:distance}
\end{proposition}}

A detailed proof is provided in Appendix~\ref{App:distance_base_prove}.  The proposition reduces the evaluation of GTME for an arbitrary $N$-mode pure Gaussian state to an optimization of the purity of a Gaussian state parameterized by $2N$ real parameters, implying linear scaling of the optimization dimension with $N$. 
The identification with purity is established by noticing that the purity of any Gaussian state $\varrho_G$ with covariance matrix $V$ is given by Tr$\varrho_G^2 = 1/\sqrt{\text{det}\, V} = \Pi_k ~\nu^{-1}_k$, with $\nu_k$s being the symplectic eigenvalues of $V$.
This structural simplification renders the numerical computation of TME efficient, which, in turn, allows us to uncover quantitative connections with the $N$-ergotropic score for multimode states.  



For any typical pure Gaussian states of three modes $\rho_3$, we numerically verify that like genuine multimode entanglement, there is again a  one-to-one relation between $\mathcal{E}^{\mathcal{G}}_3(\rho_3)$ and $\Delta^3(\rho_3)$ (see Fig. \ref{fig:relation}(b)) although this relation breaks down if one considers \(N=4\) where we find states with the same value of $\mathcal{E}^{\mathcal{G}}_3$ but different values of $\Delta^3$  
 (see Appendix~\ref{app:generation_random_state} for more details, where we also provide a brief primer of Haar uniform generation of random pure Gaussian states.)

\section{conclusion}
\label{sec:conclusion}


The intimate connection between ergotropy (the work extractable via quantum operations) and the resource theory of entanglement highlights a deep interplay between thermodynamics and quantum information science, suggesting a unified framework for the design and analysis of quantum devices. Motivated by this connection, we constructed thermodynamic quantifiers based on the difference (gap) between ergotropy extracted under global and local unitary cycles and demonstrated that these quantities faithfully capture the entanglement content of multimode Gaussian states across different bipartitions. Crucially, this approach yields a systematic characterization of inseparability in pure multimode Gaussian states, even when the system is partitioned into groups containing an arbitrary number of modes.

We introduced the notion of a \(k\)-local ergotropic gap and the associated \(k\)-ergotropic score, obtained as the minimum over all possible \(k\)-partitions. We proved that both the \(2\)-local and 
\(k\)-local ergotropic score serve as faithful quantifiers of entanglement across arbitrary mode bipartitions by explicitly showing that they satisfy all essential criteria required of an entanglement measure. Furthermore, we established that the \(2\)-local ergotropic gap provides a functionally independent upper bound on the R{\'e}nyi-2 entanglement entropy, thereby directly linking an operational thermodynamic quantity to a standard entropic measure of entanglement.  For the \(k\)-local ergotropic gap, we found that while it is directly connected to both genuine and total multimode entanglement in three-mode states, this correspondence no longer holds for states involving a larger number of modes, where these notions become inequivalent. Taken together, our results establish the ergotropic gap and the corresponding ergotropic scores as operational probes of multimode Gaussian entanglement, offering a computable thermodynamic characterization of quantum correlations in pure Gaussian states.

Recent experimental advances have demonstrated that ergotropic work can be directly accessed in realistic quantum platforms. Furthermore, the resulting 
\(k\)-ergotropic score depends solely on the sum of symplectic eigenvalues, significantly simplifying its computation compared to conventional distance-based measures. This feature enables an efficient characterization of entanglement even in large multimode systems, both from theoretical and experimental point of view. It would therefore be of considerable interest to explore potential connections between this framework and other key quantum information protocols.


Recent experiments have demonstrated that ergotropic work can be directly accessed in realistic quantum platforms~\cite{Andolina2019, Rossini2020, niu2024exp, NavidElyasi2025}. The resulting \(k\)-ergotropic score depends only on the symplectic eigenvalues, providing significant ease of computability compared to distance-based measures and enabling efficient characterization of entanglement in large multimode systems. An interesting direction is to explore whether the ergotropic framework can also capture signatures of entanglement depth~\cite{sorensen2001}, particularly in intermediate \(k\)-local regimes beyond the limiting cases \(k=2\) and \(N\). A further natural extension is to generalize our approach to mixed Gaussian states, which would broaden its applicability to realistic noisy settings. 

\acknowledgements

  The authors acknowledge the use of cluster computing facility at the Harish-Chandra Research Institute.   M.S., S.M, A.P., and A.S.D. acknowledge the support from the project entitled `` Technology Vertical - Quantum Communication'' under the National Quantum Mission of the Department of Science and Technology (DST)  ( Sanction Order No. DST/QTC/NQM/QComm/2024/2 (G)).

 \bibliography{bib}

@article{Bennett2014,
   author = {Charles H. Bennett and Gilles Brassard},
   doi = {10.1016/j.tcs.2014.05.025},
   issn = {03043975},
   journal = {Theoretical Computer Science},
   month = {12},
   pages = {7-11},
   title = {Quantum cryptography: Public key distribution and coin tossing},
   volume = {560},
   year = {2014}
}

@article{Ekert1991,
  title = {Quantum cryptography based on Bell's theorem},
  author = {Ekert, Artur K.},
  journal = {Phys. Rev. Lett.},
  volume = {67},
  issue = {6},
  pages = {661--663},
  numpages = {0},
  year = {1991},
  month = {Aug},
  publisher = {American Physical Society},
  doi = {10.1103/PhysRevLett.67.661},
  url = {https://link.aps.org/doi/10.1103/PhysRevLett.67.661}
}

@article{Bennett1992,
   author = {Charles H. Bennett and François Bessette and Gilles Brassard and Louis Salvail and John Smolin},
   doi = {10.1007/BF00191318},
   issn = {0933-2790},
   issue = {1},
   journal = {Journal of Cryptology},
   month = {1},
   pages = {3-28},
   title = {Experimental quantum cryptography},
   volume = {5},
   year = {1992}
}

@article{Bouwmeester1997,
   author = {Dik Bouwmeester and Jian-Wei Pan and Klaus Mattle and Manfred Eibl and Harald Weinfurter and Anton Zeilinger},
   doi = {10.1038/37539},
   issn = {0028-0836},
   issue = {6660},
   journal = {Nature},
   month = {12},
   pages = {575-579},
   title = {Experimental quantum teleportation},
   volume = {390},
   year = {1997}
}

@article{pan1998,
  title = {Experimental Entanglement Swapping: Entangling Photons That Never Interacted},
  author = {Pan, Jian-Wei and Bouwmeester, Dik and Weinfurter, Harald and Zeilinger, Anton},
  journal = {Phys. Rev. Lett.},
  volume = {80},
  issue = {18},
  pages = {3891--3894},
  numpages = {0},
  year = {1998},
  month = {May},
  publisher = {American Physical Society},
  doi = {10.1103/PhysRevLett.80.3891},
  url = {https://link.aps.org/doi/10.1103/PhysRevLett.80.3891}
}

@article{Gisin2002,
  title = {Quantum cryptography},
  author = {Gisin, Nicolas and Ribordy, Gr\'egoire and Tittel, Wolfgang and Zbinden, Hugo},
  journal = {Rev. Mod. Phys.},
  volume = {74},
  issue = {1},
  pages = {145--195},
  numpages = {0},
  year = {2002},
  month = {Mar},
  publisher = {American Physical Society},
  doi = {10.1103/RevModPhys.74.145},
  url = {https://link.aps.org/doi/10.1103/RevModPhys.74.145}
}

@article{Kimble2008,
   author = {H. J. Kimble},
   doi = {10.1038/nature07127},
   issn = {0028-0836},
   issue = {7198},
   journal = {Nature},
   month = {6},
   pages = {1023-1030},
   title = {The quantum internet},
   volume = {453},
   year = {2008}
}

@article{pan2012rmp,
  title = {Multiphoton entanglement and interferometry},
  author = {Pan, Jian-Wei and Chen, Zeng-Bing and Lu, Chao-Yang and Weinfurter, Harald and Zeilinger, Anton and \ifmmode \dot{Z}\else \.{Z}\fi{}ukowski, Marek},
  journal = {Rev. Mod. Phys.},
  volume = {84},
  issue = {2},
  pages = {777--838},
  numpages = {0},
  year = {2012},
  month = {May},
  publisher = {American Physical Society},
  doi = {10.1103/RevModPhys.84.777},
  url = {https://link.aps.org/doi/10.1103/RevModPhys.84.777}
}

@article{Flamini2019,
   author = {Fulvio Flamini and Nicolò Spagnolo and Fabio Sciarrino},
   doi = {10.1088/1361-6633/aad5b2},
   issn = {0034-4885},
   issue = {1},
   journal = {Reports on Progress in Physics},
   month = {1},
   pages = {016001},
   title = {Photonic quantum information processing: a review},
   volume = {82},
   year = {2019}
}

@article{Zhong2020,
   author = {Han-Sen Zhong and Hui Wang and Yu-Hao Deng and Ming-Cheng Chen and Li-Chao Peng and Yi-Han Luo and Jian Qin and Dian Wu and Xing Ding and Yi Hu and Peng Hu and Xiao-Yan Yang and Wei-Jun Zhang and Hao Li and Yuxuan Li and Xiao Jiang and Lin Gan and Guangwen Yang and Lixing You and Zhen Wang and Li Li and Nai-Le Liu and Chao-Yang Lu and Jian-Wei Pan},
   doi = {10.1126/science.abe8770},
   issn = {0036-8075},
   issue = {6523},
   journal = {Science},
   month = {12},
   pages = {1460-1463},
   title = {Quantum computational advantage using photons},
   volume = {370},
   year = {2020}
}

@article{Pirandola2020,
   author = {S. Pirandola and U. L. Andersen and L. Banchi and M. Berta and D. Bunandar and R. Colbeck and D. Englund and T. Gehring and C. Lupo and C. Ottaviani and J. L. Pereira and M. Razavi and J. Shamsul Shaari and M. Tomamichel and V. C. Usenko and G. Vallone and P. Villoresi and P. Wallden},
   doi = {10.1364/AOP.361502},
   issn = {1943-8206},
   issue = {4},
   journal = {Advances in Optics and Photonics},
   month = {12},
   pages = {1012},
   title = {Advances in quantum cryptography},
   volume = {12},
   year = {2020}
}

@article{Allen1992,
  title = {Orbital angular momentum of light and the transformation of Laguerre-Gaussian laser modes},
  author = {Allen, L. and Beijersbergen, M. W. and Spreeuw, R. J. C. and Woerdman, J. P.},
  journal = {Phys. Rev. A},
  volume = {45},
  issue = {11},
  pages = {8185--8189},
  numpages = {0},
  year = {1992},
  month = {Jun},
  publisher = {American Physical Society},
  doi = {10.1103/PhysRevA.45.8185},
  url = {https://link.aps.org/doi/10.1103/PhysRevA.45.8185}
}

@article{Kwiat1995,
   author = {Paul G. Kwiat and Klaus Mattle and Harald Weinfurter and Anton Zeilinger and Alexander V. Sergienko and Yanhua Shih},
   doi = {10.1103/PhysRevLett.75.4337},
   issn = {0031-9007},
   issue = {24},
   journal = {Physical Review Letters},
   month = {12},
   pages = {4337-4341},
   title = {New High-Intensity Source of Polarization-Entangled Photon Pairs},
   volume = {75},
   year = {1995}
}

@article{Mair2001,
   author = {Alois Mair and Alipasha Vaziri and Gregor Weihs and Anton Zeilinger},
   doi = {10.1038/35085529},
   issn = {0028-0836},
   issue = {6844},
   journal = {Nature},
   month = {7},
   pages = {313-316},
   title = {Entanglement of the orbital angular momentum states of photons},
   volume = {412},
   year = {2001}
}

@article{Braunstein2005,
  title = {Quantum information with continuous variables},
  author = {Braunstein, Samuel L. and van Loock, Peter},
  journal = {Rev. Mod. Phys.},
  volume = {77},
  issue = {2},
  pages = {513--577},
  numpages = {0},
  year = {2005},
  month = {Jun},
  publisher = {American Physical Society},
  doi = {10.1103/RevModPhys.77.513},
  url = {https://link.aps.org/doi/10.1103/RevModPhys.77.513}
}

@article{OBrien2009,
   author = {Jeremy L. O'Brien and Akira Furusawa and Jelena Vučković},
   doi = {10.1038/nphoton.2009.229},
   issn = {1749-4885},
   issue = {12},
   journal = {Nature Photonics},
   month = {12},
   pages = {687-695},
   title = {Photonic quantum technologies},
   volume = {3},
   year = {2009}
}

@article{Weedbrook2012rmp,
  title = {Gaussian quantum information},
  author = {Weedbrook, Christian and Pirandola, Stefano and Garc\'{\i}a-Patr\'on, Ra\'ul and Cerf, Nicolas J. and Ralph, Timothy C. and Shapiro, Jeffrey H. and Lloyd, Seth},
  journal = {Rev. Mod. Phys.},
  volume = {84},
  issue = {2},
  pages = {621--669},
  numpages = {0},
  year = {2012},
  month = {May},
  publisher = {American Physical Society},
  doi = {10.1103/RevModPhys.84.621},
  url = {https://link.aps.org/doi/10.1103/RevModPhys.84.621}
}

@article{GRMP,
  title = {Gaussian quantum information},
  author = {Weedbrook, Christian and Pirandola, Stefano and Garc\'{\i}a-Patr\'on, Ra\'ul and Cerf, Nicolas J. and Ralph, Timothy C. and Shapiro, Jeffrey H. and Lloyd, Seth},
  journal = {Rev. Mod. Phys.},
  volume = {84},
  issue = {2},
  pages = {621--669},
  numpages = {0},
  year = {2012},
  month = {May},
  publisher = {American Physical Society},
  doi = {10.1103/RevModPhys.84.621},
  url = {https://link.aps.org/doi/10.1103/RevModPhys.84.621}
}

@article{cvtele,
  title = {Teleportation of Continuous Quantum Variables},
  author = {Braunstein, Samuel L. and Kimble, H. J.},
  journal = {Phys. Rev. Lett.},
  volume = {80},
  issue = {4},
  pages = {869--872},
  numpages = {0},
  year = {1998},
  month = {Jan},
  publisher = {American Physical Society},
  doi = {10.1103/PhysRevLett.80.869},
  url = {https://link.aps.org/doi/10.1103/PhysRevLett.80.869}
}

@article{cvdc,
  title = {Dense coding for continuous variables},
  author = {Braunstein, Samuel L. and Kimble, H. J.},
  journal = {Phys. Rev. A},
  volume = {61},
  issue = {4},
  pages = {042302},
  numpages = {4},
  year = {2000},
  month = {Mar},
  publisher = {American Physical Society},
  doi = {10.1103/PhysRevA.61.042302},
  url = {https://link.aps.org/doi/10.1103/PhysRevA.61.042302}
}

@article{Zhang2024,
  title = {Continuous-variable quantum key distribution system: Past,  present,  and future},
  volume = {11},
  ISSN = {1931-9401},
  url = {http://dx.doi.org/10.1063/5.0179566},
  DOI = {10.1063/5.0179566},
  number = {1},
  journal = {Applied Physics Reviews},
  publisher = {AIP Publishing},
  author = {Zhang,  Yichen and Bian,  Yiming and Li,  Zhengyu and Yu,  Song and Guo,  Hong},
  year = {2024},
  month = mar 
}

@article{cvec1,
  author = {Gottesman, Daniel and Kitaev, Alexei and Preskill, John},
  title = {Encoding a qubit in an oscillator},
  journal = {Phys. Rev. A},
  volume = {64},
  pages = {012310},
  year = {2001},
  doi = {10.1103/PhysRevA.64.012310}
}

@article{cvec2,
  author = {Niset, Julien and Fiur\'a\v{s}ek, Jarom\'{\i}r and Cerf, Nicolas J.},
  title = {No-Go Theorem for Gaussian Quantum Error Correction},
  journal = {Phys. Rev. Lett.},
  volume = {102},
  pages = {120501},
  year = {2009},
  doi = {10.1103/PhysRevLett.102.120501}
}

@article{cvec3,
  author = {Ofek, Nissim and Petrenko, Andrei and Heeres, Reinier W. and Reinhold, Philipp and Leghtas, Zaki and Vlastakis, Brian and Liu, Yaxing and Frunzio, Luigi and Girvin, S. M. and Jiang, Liang and Mirrahimi, Mazyar and Devoret, Michel H. and Schoelkopf, Robert J.},
  title = {Extending the lifetime of a quantum bit with error correction in superconducting circuits},
  journal = {Nature},
  volume = {536},
  pages = {441--445},
  year = {2016},
  doi = {10.1038/nature18949}
}

@article{cvec4,
  author = {Michael, M. H. and Silveri, M. and Brierley, R. T. and Albert, V. V. and Salmilehto, J. and Jiang, Liang and Girvin, S. M.},
  title = {New Class of Quantum Error-Correcting Codes for a Bosonic Mode},
  journal = {Phys. Rev. X},
  volume = {6},
  pages = {031006},
  year = {2016},
  doi = {10.1103/PhysRevX.6.031006}
}

@article{cvec5,
  author = {Albert, Victor V. and Shen, Chang-Ling and Brierley, R. T. and Verstraete, Frank and Jiang, Liang},
  title = {Performance and structure of single-mode bosonic codes},
  journal = {Phys. Rev. A},
  volume = {97},
  pages = {032346},
  year = {2018},
  doi = {10.1103/PhysRevA.97.032346}
}

@article{cvec6,
  author = {Campagne-Ibarcq, Philippe and Eickbusch, Arne and Touzard, Simon and Zalys-Geller, Eran and Frattini, Nicholas E. and Sivak, V. V. and Reinhold, Philipp and Puri, Shruti and Shankar, S. and Devoret, Michel H. and Girvin, S. M. and Mirrahimi, Mazyar and Schoelkopf, Robert J.},
  title = {Quantum error correction of a qubit encoded in grid states of an oscillator},
  journal = {Nature},
  volume = {584},
  pages = {368--372},
  year = {2020},
  doi = {10.1038/s41586-020-2603-3}
}

@article{cvec7,
  author = {Menicucci, Nicolas C.},
  title = {Fault-Tolerant Measurement-Based Quantum Computing with Continuous-Variable Cluster States},
  journal = {Phys. Rev. Lett.},
  volume = {112},
  pages = {120504},
  year = {2014},
  doi = {10.1103/PhysRevLett.112.120504}
}

@article{cvmbqc1,
  author = {Menicucci, Nicolas C. and Van Loock, Peter and Gu, Mile and Weedbrook, Christian and Ralph, Timothy C. and Nielsen, Michael A.},
  title = {Universal quantum computation with continuous-variable cluster states},
  journal = {Phys. Rev. Lett.},
  volume = {97},
  pages = {110501},
  year = {2006},
  doi = {10.1103/PhysRevLett.97.110501}
}

@article{cvmbqc2,
  author = {Gu, Mile and Weedbrook, Christian and Menicucci, Nicolas C. and Ralph, Timothy C. and Van Loock, Peter},
  title = {Quantum computing with continuous-variable clusters},
  journal = {Phys. Rev. A},
  volume = {79},
  pages = {062318},
  year = {2009},
  doi = {10.1103/PhysRevA.79.062318}
}

@article{cvmbqc4,
  author = {Du, Peilin and Zhang, Jing and Zhang, Tiancai and Yang, Rongguo and Gao, Jiangrui},
  title = {A complete continuous-variable quantum computation architecture based on the 2D spatiotemporal cluster state},
  journal = {Sci. Rep.},
  volume = {15},
  pages = {18199},
  year = {2025},
  doi = {10.1038/s41598-025-02899-8}
}

@article{Lee2025,
  title = {Photonic hybrid quantum computing},
  ISSN = {2950-6360},
  url = {http://dx.doi.org/10.1016/j.newton.2025.100359},
  DOI = {10.1016/j.newton.2025.100359},
  journal = {Newton},
  publisher = {Elsevier BV},
  author = {Lee,  Jaehak and Omkar,  Srikrishna and Teo,  Yong Siah and Lee,  Seok-Hyung and Kwon,  Hyukjoon and Kim,  M.S. and Jeong,  Hyunseok},
  year = {2025},
  month = dec,
  pages = {100359}
}

@article{landauer1961,
  author  = {Landauer, R.},
  title   = {Irreversibility and Heat Generation in the Computing Process},
  journal = {IBM J. Res. Dev.},
  volume  = {5},
  pages   = {183--191},
  year    = {1961},
  doi     = {10.1147/rd.53.0183}
}

@article{bennett1982,
  author  = {Bennett, Charles H.},
  title   = {The Thermodynamics of Computation---A Review},
  journal = {Int. J. Theor. Phys.},
  volume  = {21},
  pages   = {905--940},
  year    = {1982},
  doi     = {10.1007/BF02084158}
}

@incollection{zurek1986,
  author    = {Zurek, W. H.},
  title     = {Maxwell's Demon, Szilard's Engine and Quantum Measurements},
  booktitle = {Frontiers of Nonequilibrium Statistical Physics},
  editor    = {Moore, G. T. and Scully, M. O.},
  publisher = {Plenum Press},
  address   = {New York},
  pages     = {151--161},
  year      = {1986},
  doi       = {10.1007/978-1-4613-2219-9_12}
}

@article{toyabe2010,
  author  = {Toyabe, Shoichi and Sagawa, Takahiro and Ueda, Masahito and Muneyuki, Eiro and Sano, Masaki},
  title   = {Experimental demonstration of information-to-energy conversion and validation of the generalized Jarzynski equality},
  journal = {Nat. Phys.},
  volume  = {6},
  pages   = {988--992},
  year    = {2010},
  doi     = {10.1038/nphys1821}
}

@article{berut2012,
  author  = {B{\'e}rut, Antoine and Arakelyan, Artak and Petrosyan, Artyom and Ciliberto, Sergio and Dillenschneider, Raoul and Lutz, Eric},
  title   = {Experimental verification of Landauer’s principle linking information and thermodynamics},
  journal = {Nature},
  volume  = {483},
  pages   = {187--189},
  year    = {2012},
  doi     = {10.1038/nature10872}
}

@article{horodecki2013limits,
  author  = {Horodecki, Micha{\l} and Oppenheim, Jonathan},
  title   = {Fundamental limitations for quantum and nanoscale thermodynamics},
  journal = {Nat. Commun.},
  volume  = {4},
  pages   = {2059},
  year    = {2013},
  doi     = {10.1038/ncomms3059}
}

@article{brandao2013resource,
  author  = {Brand\~{a}o, F. G. S. L. and Horodecki, Micha{\l} and Oppenheim, Jonathan and Renes, Joseph M. and Spekkens, Robert W.},
  title   = {The Resource Theory of Quantum States Out of Thermal Equilibrium},
  journal = {Phys. Rev. Lett.},
  volume  = {111},
  pages   = {250404},
  year    = {2013},
  doi     = {10.1103/PhysRevLett.111.250404}
}

@article{brandao2015secondlaws,
  author  = {Brand\~{a}o, F. G. S. L. and Horodecki, Micha{\l} and Ng, N. H. Y. and Oppenheim, Jonathan and Wehner, Stephanie},
  title   = {The Second Laws of Quantum Thermodynamics},
  journal = {Proc. Natl. Acad. Sci. U.S.A.},
  volume  = {112},
  pages   = {3275--3279},
  year    = {2015},
  doi     = {10.1073/pnas.1411728112}
}

@article{egap1,
  title = {Extractable Work from Correlations},
  author = {Perarnau-Llobet, Mart\'{\i} and Hovhannisyan, Karen V. and Huber, Marcus and Skrzypczyk, Paul and Brunner, Nicolas and Ac\'{\i}n, Antonio},
  journal = {Phys. Rev. X},
  volume = {5},
  issue = {4},
  pages = {041011},
  numpages = {14},
  year = {2015},
  month = {Oct},
  publisher = {American Physical Society},
  doi = {10.1103/PhysRevX.5.041011},
  url = {https://link.aps.org/doi/10.1103/PhysRevX.5.041011}
}

@article{egap2,
  title = {Bound on ergotropic gap for bipartite separable states},
  author = {Alimuddin, Mir and Guha, Tamal and Parashar, Preeti},
  journal = {Phys. Rev. A},
  volume = {99},
  issue = {5},
  pages = {052320},
  numpages = {17},
  year = {2019},
  month = {May},
  publisher = {American Physical Society},
  doi = {10.1103/PhysRevA.99.052320},
  url = {https://link.aps.org/doi/10.1103/PhysRevA.99.052320}
}

@article{comrmp,
  title = {Continuous-variable quantum communication},
  author = {Usenko, Vladyslav C. and Acín, Antonio and Alléaume, Romain and Andersen, Ulrik L. and Diamanti, Eleni and Gehring, Tobias and Hajomer, Adnan A. E. and Kanitschar, Florian and Pacher, Christoph and Pirandola, Stefano and Pruneri, Valerio},
  journal = {Rev. Mod. Phys.},
  pages = {--},
  year = {2025},
  month = {Nov},
  publisher = {American Physical Society},
  doi = {10.1103/mgj7-t6d3},
  url = {https://link.aps.org/doi/10.1103/mgj7-t6d3}
}

@article{Polo2025,
  title = {Ergotropic Characterization of Continuous-Variable Entanglement},
  author = {Polo-Rodr\'{\i}guez, Beatriz and Centrone, Federico and Adesso, Gerardo and Alimuddin, Mir},
  journal = {Phys. Rev. Lett.},
  volume = {136},
  issue = {5},
  pages = {050201},
  numpages = {7},
  year = {2026},
  month = {Feb},
  publisher = {American Physical Society},
  doi = {10.1103/43jm-qkz2},
  url = {https://link.aps.org/doi/10.1103/43jm-qkz2}
}

@article{Botero03,
  title = {Modewise entanglement of Gaussian states},
  author = {Botero, Alonso and Reznik, Benni},
  journal = {Phys. Rev. A},
  volume = {67},
  issue = {5},
  pages = {052311},
  numpages = {5},
  year = {2003},
  month = {May},
  publisher = {American Physical Society},
  doi = {10.1103/PhysRevA.67.052311},
  url = {https://link.aps.org/doi/10.1103/PhysRevA.67.052311}
}

@article{Fiur02,
  title = {Gaussian Transformations and Distillation of Entangled Gaussian States},
  author = {Fiur\'a\ifmmode \check{s}\else \v{s}\fi{}ek, Jarom\'{\i}r},
  journal = {Phys. Rev. Lett.},
  volume = {89},
  issue = {13},
  pages = {137904},
  numpages = {4},
  year = {2002},
  month = {Sep},
  publisher = {American Physical Society},
  doi = {10.1103/PhysRevLett.89.137904},
  url = {https://link.aps.org/doi/10.1103/PhysRevLett.89.137904}
}

@article{Puliyil22,
  title = {Thermodynamic Signatures of Genuinely Multipartite Entanglement},
  author = {Puliyil, Samgeeth and Banik, Manik and Alimuddin, Mir},
  journal = {Phys. Rev. Lett.},
  volume = {129},
  issue = {7},
  pages = {070601},
  numpages = {7},
  year = {2022},
  month = {Aug},
  publisher = {American Physical Society},
  doi = {10.1103/PhysRevLett.129.070601},
  url = {https://link.aps.org/doi/10.1103/PhysRevLett.129.070601}
}

@misc{Giedke03,
      title={Entanglement transformations of pure Gaussian states}, 
      author={G. Giedke and J. Eisert and J. I. Cirac and M. B. Plenio},
      year={2003},
      eprint={quant-ph/0301038},
      archivePrefix={arXiv},
      primaryClass={quant-ph},
      url={https://arxiv.org/abs/quant-ph/0301038}, 
}

@article{qc1,
  title = {Concentrating partial entanglement by local operations},
  author = {Bennett, Charles H. and Bernstein, Herbert J. and Popescu, Sandu and Schumacher, Benjamin},
  journal = {Phys. Rev. A},
  volume = {53},
  issue = {4},
  pages = {2046--2052},
  numpages = {0},
  year = {1996},
  month = {Apr},
  publisher = {American Physical Society},
  doi = {10.1103/PhysRevA.53.2046},
  url = {https://link.aps.org/doi/10.1103/PhysRevA.53.2046}
}

@article{qc2,
  title = {Quantum entanglement},
  author = {Horodecki, Ryszard and Horodecki, Pawe\l{} and Horodecki, Micha\l{} and Horodecki, Karol},
  journal = {Rev. Mod. Phys.},
  volume = {81},
  issue = {2},
  pages = {865--942},
  numpages = {0},
  year = {2009},
  month = {Jun},
  publisher = {American Physical Society},
  doi = {10.1103/RevModPhys.81.865},
  url = {https://link.aps.org/doi/10.1103/RevModPhys.81.865}
}

@article{coma,
  title = {Significance of fidelity deviation in continuous-variable teleportation},
  author = {Patra, Ayan and Gupta, Rivu and Roy, Saptarshi and Sen(De), Aditi},
  journal = {Phys. Rev. A},
  volume = {106},
  issue = {2},
  pages = {022433},
  numpages = {15},
  year = {2022},
  month = {Aug},
  publisher = {American Physical Society},
  doi = {10.1103/PhysRevA.106.022433},
  url = {https://link.aps.org/doi/10.1103/PhysRevA.106.022433}
}

@article{comb,
  title = {Quantum dense coding network using multimode squeezed states of light},
  author = {Patra, Ayan and Gupta, Rivu and Roy, Saptarshi and Das, Tamoghna and Sen(De), Aditi},
  journal = {Phys. Rev. A},
  volume = {106},
  issue = {5},
  pages = {052607},
  numpages = {10},
  year = {2022},
  month = {Nov},
  publisher = {American Physical Society},
  doi = {10.1103/PhysRevA.106.052607},
  url = {https://link.aps.org/doi/10.1103/PhysRevA.106.052607}
}

@article{good,
  title = {Limits for Entanglement Measures},
  author = {Horodecki, Micha\l{} and Horodecki, Pawe\l{} and Horodecki, Ryszard},
  journal = {Phys. Rev. Lett.},
  volume = {84},
  issue = {9},
  pages = {2014--2017},
  numpages = {0},
  year = {2000},
  month = {Feb},
  publisher = {American Physical Society},
  doi = {10.1103/PhysRevLett.84.2014},
  url = {https://link.aps.org/doi/10.1103/PhysRevLett.84.2014}
}

@article{Adesso2014,
  title = {Continuous Variable Quantum Information: Gaussian States and Beyond},
  volume = {21},
  ISSN = {1793-7191},
  url = {http://dx.doi.org/10.1142/S1230161214400010},
  DOI = {10.1142/s1230161214400010},
  number = {01n02},
  journal = {Open Syst. Inf. Dyn.},
  publisher = {World Scientific Pub Co Pte Lt},
  author = {Adesso,  Gerardo and Ragy,  Sammy and Lee,  Antony R.},
  year = {2014},
  month = mar,
  pages = {1440001}
}

@incollection{abramowitz1972stirling,
  author    = {Abramowitz, Milton and Stegun, Irene A.},
  title     = {Stirling Numbers of the Second Kind},
  booktitle = {Handbook of Mathematical Functions with Formulas, Graphs, and Mathematical Tables},
  publisher = {Dover Publications},
  address   = {New York},
  year      = {1972},
  pages     = {824--825},
  edition   = {9th printing},
  note      = {Section 24.1.4}
}

@article{strong,
  title = {Measuring Gaussian Quantum Information and Correlations Using the R\'enyi Entropy of Order 2},
  author = {Adesso, Gerardo and Girolami, Davide and Serafini, Alessio},
  journal = {Phys. Rev. Lett.},
  volume = {109},
  issue = {19},
  pages = {190502},
  numpages = {6},
  year = {2012},
  month = {Nov},
  publisher = {American Physical Society},
  doi = {10.1103/PhysRevLett.109.190502},
  url = {https://link.aps.org/doi/10.1103/PhysRevLett.109.190502}
}

@article{ggm1,
  title = {Channel capacities versus entanglement measures in multiparty quantum states},
  author = {Sen(De), Aditi and Sen, Ujjwal},
  journal = {Phys. Rev. A},
  volume = {81},
  issue = {1},
  pages = {012308},
  numpages = {6},
  year = {2010},
  month = {Jan},
  publisher = {American Physical Society},
  doi = {10.1103/PhysRevA.81.012308},
  url = {https://link.aps.org/doi/10.1103/PhysRevA.81.012308}
}

@article{ggm2,
  title = {Generalized geometric measure of entanglement for multiparty mixed states},
  author = {Das, Tamoghna and Roy, Sudipto Singha and Bagchi, Shrobona and Misra, Avijit and Sen(De), Aditi and Sen, Ujjwal},
  journal = {Phys. Rev. A},
  volume = {94},
  issue = {2},
  pages = {022336},
  numpages = {8},
  year = {2016},
  month = {Aug},
  publisher = {American Physical Society},
  doi = {10.1103/PhysRevA.94.022336},
  url = {https://link.aps.org/doi/10.1103/PhysRevA.94.022336}
}

@article{Spedalieri2014,
   author = {Gaetana Spedalieri and Christian Weedbrook and Stefano Pirandola},
   doi = {10.1088/1751-8113/47/32/329501},
   issn = {1751-8113},
   issue = {32},
   journal = {Journal of Physics A: Mathematical and Theoretical},
   month = {8},
   pages = {329501},
   title = {Corrigendum: A limit formula for the quantum fidelity (2013 J. Phys. A: Math. Theor. 46 025304)},
   volume = {47},
   year = {2014}
}

@article{Saptarshi20,
  title = {Computable genuine multimode entanglement measure: Gaussian versus non-Gaussian},
  author = {Roy, Saptarshi and Das, Tamoghna and Sen(De), Aditi},
  journal = {Phys. Rev. A},
  volume = {102},
  issue = {1},
  pages = {012421},
  numpages = {13},
  year = {2020},
  month = {Jul},
  publisher = {American Physical Society},
  doi = {10.1103/PhysRevA.102.012421},
  url = {https://link.aps.org/doi/10.1103/PhysRevA.102.012421}
}

@article{Saptarshi25,
  title = {Typical behavior of genuine multimode entanglement of pure Gaussian states},
  author = {Roy, Saptarshi},
  journal = {Phys. Rev. A},
  volume = {111},
  issue = {6},
  pages = {062414},
  numpages = {8},
  year = {2025},
  month = {Jun},
  publisher = {American Physical Society},
  doi = {10.1103/8bbb-jhzh},
  url = {https://link.aps.org/doi/10.1103/8bbb-jhzh}
}

@book{BraunsteinPati2003,
  editor    = {S. L. Braunstein and A. K. Pati},
  title     = {Quantum Continuous Variables: A Primer of Theoretical Methods},
  publisher = {Springer},
  address   = {Dordrecht},
  year      = {2003},
  doi       = {10.1007/978-94-017-0343-4}
}

@book{CerfLeuchsPolzik2007,
  editor    = {N. J. Cerf and G. Leuchs and E. S. Polzik},
  title     = {Quantum Information with Continuous Variables of Atoms and Light},
  publisher = {Imperial College Press},
  address   = {London},
  year      = {2007},
  doi       = {10.1142/p489}
}

@article{Huber2015,
  author  = {Marcus Huber and Martí Perarnau-Llobet and Julio I. de Vicente},
  title   = {Thermodynamic cost of creating correlations},
  journal = {New J. Phys.},
  volume  = {17},
  pages   = {065008},
  year    = {2015},
  doi     = {10.1088/1367-2630/17/6/065008}
}

@article{Sapienza2019,
  author  = {G. Sapienza and M. Huber and J. Oppenheim and A. Winter},
  title   = {Thermodynamic signatures of quantum correlations},
  journal = {Nat. Commun.},
  volume  = {10},
  pages   = {2492},
  year    = {2019},
  doi     = {10.1038/s41467-019-10572-8}
}

@article{Allahverdyan2004,
  author  = {A. E. Allahverdyan and R. Balian and Theo M. Nieuwenhuizen},
  title   = {Maximal work extraction from finite quantum systems},
  journal = {Europhys. Lett.},
  volume  = {67},
  pages   = {565--571},
  year    = {2004},
  doi     = {10.1209/epl/i2004-10101-2}
}

@article{Cao2025,
  author  = {Y. Cao and others},
  title   = {Ergotropic gap and multipartite correlations},
  journal = {Phys. Lett. A},
  volume  = {559},
  pages   = {130908},
  year    = {2025},
  doi     = {10.1016/j.physleta.2025.130908}
}

@article{sorensen2001,
  title = {Entanglement and Extreme Spin Squeezing},
  author = {S\o{}rensen, Anders S. and M\o{}lmer, Klaus},
  journal = {Phys. Rev. Lett.},
  volume = {86},
  issue = {20},
  pages = {4431--4434},
  numpages = {0},
  year = {2001},
  month = {May},
  publisher = {American Physical Society},
  doi = {10.1103/PhysRevLett.86.4431},
  url = {https://link.aps.org/doi/10.1103/PhysRevLett.86.4431}
}

@misc{niu2024exp,
      title={Experimental investigation of coherent ergotropy in a single spin system}, 
      author={Zhibo Niu and Yang Wu and Yunhan Wang and Xing Rong and Jiangfeng Du},
      year={2024},
      eprint={2409.06249},
      archivePrefix={arXiv},
      primaryClass={quant-ph},
      url={https://arxiv.org/abs/2409.06249}, 
}

@article{Andolina2019,
  title = {Extractable Work, the Role of Correlations, and Asymptotic Freedom in Quantum Batteries},
  author = {Andolina, Gian Marcello and Keck, Maximilian and Mari, Andrea and Campisi, Michele and Giovannetti, Vittorio and Polini, Marco},
  journal = {Phys. Rev. Lett.},
  volume = {122},
  issue = {4},
  pages = {047702},
  numpages = {5},
  year = {2019},
  month = {Feb},
  publisher = {American Physical Society},
  doi = {10.1103/PhysRevLett.122.047702},
  url = {https://link.aps.org/doi/10.1103/PhysRevLett.122.047702}
}

@article{Rossini2020,
  title = {Quantum Advantage in the Charging Process of Sachdev-Ye-Kitaev Batteries},
  author = {Rossini, Davide and Andolina, Gian Marcello and Rosa, Dario and Carrega, Matteo and Polini, Marco},
  journal = {Phys. Rev. Lett.},
  volume = {125},
  issue = {23},
  pages = {236402},
  numpages = {6},
  year = {2020},
  month = {Dec},
  publisher = {American Physical Society},
  doi = {10.1103/PhysRevLett.125.236402},
  url = {https://link.aps.org/doi/10.1103/PhysRevLett.125.236402}
}

@article{NavidElyasi2025,
   author = {Seyed Navid Elyasi and Matteo A C Rossi and Marco G Genoni},
   doi = {10.1088/2058-9565/adae2d},
   issn = {2058-9565},
   issue = {2},
   journal = {Quantum Science and Technology},
   month = {4},
   pages = {025017},
   title = {Experimental simulation of daemonic work extraction in open quantum batteries on a digital quantum computer},
   volume = {10},
   year = {2025}
}

@incollection{EisertWolf2007,
  author       = {Jens Eisert and Michael M. Wolf},
  title        = {Gaussian Quantum Channels},
  booktitle    = {Quantum Information with Continuous Variables of Atoms and Light},
  editor       = {Nicolas J. Cerf and Gerd Leuchs and Eugene S. Polzik},
  publisher    = {Imperial College Press},
  address      = {London, UK},
  year         = {2007},
  pages        = {23--42},
  isbn         = {978-1-86094-741-9},
}

\widetext

\appendix

\section{Preliminaries}
\label{sec:pre}

In this section, we first review the phase-space description of Gaussian states and briefly discuss the entanglement measure employed to characterize genuine multimode correlations.

\subsection{Continuous variable formalism}
\label{sec:cv_form}

Let us consider a system of \(N\) bosonic modes described within the continuous-variable framework, whose total Hamiltonian is assumed to be a sum of independent mode Hamiltonians and can be written as
\begin{equation}
\hat{H} = \sum_{i=1}^{N} \hat{H}_i , 
\label{equ:hamiltonian}
\end{equation}
where each individual mode is governed by the general second order Hamiltonian
\begin{equation}
\hat{H}_i = \frac12 (\bm{\hat R_i}  - \bm{r_i})^T h (\bm{\hat R_i}  - \bm{r_i}).
\label{eq:single_mod_ham}
\end{equation}
Here $\hat{\bm R}_i=(\hat q_i,\hat p_i)^{T}$ denotes the vector of canonical quadrature operators of the $i$th mode, satisfying the commutation relations $[\hat q_k,\hat p_l]=i\delta_{kl}$, with $\bm r_i\in\mathbb{R}^2$ and $h > 0 \in \mathbb{R}^2\times \mathbb{R}^2$.
 For convenience, we collect the quadrature operators of all modes into a single vector, \(\hat{R}=(\hat{q}_1,\hat{p}_1,\hat{q}_2,\hat{p}_2,.....,\hat{q}_N,\hat{p}_N)^T\). In this representation, the canonical commutation relations take a compact matrix form, \([\hat{R}_k,\hat{R}_l]=i\Omega_{kl}\), where \(\Omega\) denotes the symplectic matrix associated with an \(N\) mode bosonic system and written as 
\begin{equation}
\Omega = \bigoplus_{k=1}^{N} \omega, \qquad
\omega =
\begin{pmatrix}
0 & 1 \\
-1 & 0
\end{pmatrix}.
\end{equation}
Within this phase space formulation, any Gaussian state \(\rho_{_N}\) is completely characterized by its first moment (displacement) vector $\bm{d}$ and its second moment, i.e., covariance matrix (CM) $\Sigma$ . These quantities are defined, respectively, as
\begin{eqnarray}
   \nonumber {d}_i &=& \langle \hat{R}_i\rangle  {,~\text{and}}\\
   \Sigma_{ij}&=& \langle \hat{R}_i \hat{R}_j + \hat{R}_j R_i\rangle - 2\langle \hat{R}_i \rangle \langle \hat{R}_j \rangle .
\end{eqnarray}
The CM must satisfy the bonafide uncertainty relation \(\Sigma+i\Omega \geq 0\). 
{\color{black} In this paper, we work with a special case of the Hamiltonian in Eq. \eqref{eq:single_mod_ham}. Specifically, we consider it to be Harmonic with $h = \mathbb{I}_2$ and $\bm{r}_i = 0$ for all modes $i$. Furthermore, since we are interested in correlations, without the loss of generality, we consider states that are displacement-free, i.e. $\bm{d} = \bm{0}.$ Therefore, the average energy of an $N-$mode state $\rho_{_N}$ with covariance matrix $\Sigma$ can be expressed as}
\begin{eqnarray}
E(\rho_{_N})&=&\Tr(\rho_{_{N}}H) = \frac{1}{4}\Tr \Sigma.
    \label{eq:energy}
\end{eqnarray}

{\color{black}
\subsection{Multimode Gaussian entanglement: Bipartite vs Multipartite}
\label{subsec:entmeasure}

A multimode pure quantum state can be entangled in qualitatively different ways. One may wish to consider its entanglement content across a given bipartition. For example, the bipartite entanglement content of a pure multimode state is characterized across a given bipartition, which for a Gaussian case reduces to the von-Neumann  entropy of the marginal state, where the marginals are defined from the considered bipartition.

A multimode pure state is $k$-separable if it is separable at least across $(k-1)$-bipartitions. \ayan{Note that a $k$-separable pure state is necessarily also $(k-1)$-separable, and similarly for all lower levels of separability.} Therefore, a quantum state is said to possess $k$-inseparable multipartite entanglement if it is not $k$-separable. \ayan{Furthermore, similar to the hierarchy of $k$-separability, a $k$-inseparable state is necessarily also $(k+1)$-inseparable. On the other hand, a multimode mixed state is said to be $k$-inseparable if there exists at least one pure-state decomposition in which at least one constituent pure state is $k$-inseparable, while the remaining pure states in that decomposition may be $(k+1)$-inseparable. Moreover, a $k$-inseparable mixed state cannot admit any decomposition containing a pure state with $(k-1)$-inseparable multipartite entanglement.} One of the standard methodologies to quantify the amount of $ k$-inseparable multipartite entanglement is to consider distance-based measures. In particular, the $ k$-inseparable entanglement content of a state $\rho$ can be expressed as the minimum distance of that state from the set of $k$-separable states:
\begin{eqnarray}
    \mathcal{E}_k(\rho) = \min_{\sigma \in k-\text{sep}} D(\rho,\sigma),
    \label{eq:k_distance_measure}
\end{eqnarray}
where $D$ is some distance measure between two quantum states. In this work, we intend to compute the $ k$-inseparable multipartite entanglement content of a multimode \ayan{\emph{pure}} Gaussian state from a thermodynamic perspective. For that, we first list down the minimal requirements a measure of $ k$-inseparable multimode Gaussian entanglement must satisfy.

\begin{definition}
  A non-negative quantity $\mathcal{M}_k$ qualifies as a valid $k$-inseparable multipartite entanglement measure of Gaussian state if it satisfies the following minimal requirements:
  \begin{enumerate}
      \item \emph{Vanishing for $k$-separable states.} $\mathcal{M}_k(\rho_{k-{\tt sep}})=0$ for any $k$-separable state $\rho_{k-{\tt sep}}.$ 

      \item \emph{Monotonicity under Gaussian local operations and classical communication (GLOCC), i.e., \(\mathcal{M}_k\) does not increase under deterministic GLOCC. }
  \end{enumerate}
  \label{def:requirements}
\end{definition}
Apart from these essential requirements, the measure may additionally satisfy many ``good" \cite{good} properties like faithfulness, strong-monotonicity, additivity, convexity, or may offer advantages like ease of computability \cite{good, qc2}.

\section{Proof of Lemma \ref{le:lemma_1}}
\label{app:coro_1}

Let us consider an \(N\)-mode pure Gaussian state partitioned into two subsystems, \(\mathcal{A}_1\) and \(\mathcal{A}_2\), forming a bipartition. A subsystem \(\mathcal{A}_1\) contains \(m_1\) modes and \(\mathcal{A}_2\) contains \(m_2\) modes, where \(m_1\leq m_2\) and \(m_1 + m_2 = N\). The corresponding \(2\)-local ergotropic gap, as obtained in Eq.~(\ref{eq:general_erg_gap_fixed}), is
\begin{eqnarray}
\delta^{\mathcal{A}_1|\mathcal{A}_2}_{2-\tt loc}(\rho_{_N})=\frac{1}{2}
\big( \sum_{i=1}^{m_1} \nu_i^{\mathcal{A}_1} + \sum_{i=1}^{m_2} \nu_i^{\mathcal{A}_2}- \sum_{k=1}^{N} \mu_k\big),
    \label{eq:1}
\end{eqnarray}
 where, \(\{\mu_k\}_{k=1}^{N}\) denotes the symplectic eigenvalues of \(\Sigma^{^\text{g,p}}\) as discussed in Eq.~(\ref{eq:global_eg_f}) and \(\{\nu_i^{\mathcal{A}_j}\}_{i=1}^{m_j}\) is the symplectic eigenvalues of $\Sigma_{\mathcal{A}_j}$ appeared in Eq.~\eqref{eq:gen_cova}.
 In case of \(N\)-mode pure Gaussian states, $\mu_k=1~\forall k$, while
the symplectic spectra of \(\Sigma_{\mathcal{A}_1} \text{and}~\Sigma_{\mathcal{A}_2}\) satisfy \(\nu_i^{\mathcal{A}_1} = \nu_{j=i}^{\mathcal{A}_2} ~ \forall i\in\{1,\ldots, m_1\},~ \text{and}~\nu_j^{\mathcal{A}_2}=1, \forall j\in\{m_1+1,\ldots,m_2\}\). Therefore, for \(N\)-mode pure Gaussian states, Eq.~(\ref{eq:1}) reduces to 
\begin{eqnarray}
\delta^{\mathcal{A}_1|\mathcal{A}_2}_{2\text{-loc}}(\rho_{_N})&=& \frac{1}{2}
\bigg[\sum_{i=1}^{m_1}\nu_i^{\mathcal{A}_1}+ \sum_{i=1}^{m_2}\nu_i^{\mathcal{A}_2} -N\bigg] \nonumber \\
&=&\frac{1}{2}
\bigg[\sum_{i=1}^{m_1}\nu_i^{\mathcal{A}_1}+ \bigg(\sum_{i=1}^{m_1}\nu_i^{\mathcal{A}_1} + (m_2-m_1)\bigg) -(m_1+m_2)\bigg] \nonumber \\
&=&\sum_{i=1}^{m_1} (\nu_i^{\mathcal{A}_1} - 1).
\label{equ:2}
\end{eqnarray}
This completes the proof.

\subsection{Proof of Theorem \ref{th:th_1}}
\label{app:theory_1}

To be a valid entanglement measure, the ergotropic gap, \(\delta^{\mathcal{A}_1|\mathcal{A}_2}_{2\text{-loc}}(\rho_{_N})\) should satisfy the following properties. They are as follows
\label{app:theory_1_main}

\textbf{(i) Positivity:}
The \(2\)-local ergotropic gap for \(N\)-modes pure state is 
\begin{eqnarray}
\delta^{\mathcal{A}_1|\mathcal{A}_2}_{2\text{-loc}}(\rho_{_N}) 
&=&
\sum_{i=1}^{m_1}\big(\nu_i^{\mathcal{A}_1}-1\big),
\label{eq:3}
\end{eqnarray}
 where, \(\nu_i^{\mathcal{A}_1} \geq 1, \forall i\). Therefore, each term is non-negative, which immediately implies
\begin{eqnarray}   
\delta^{\mathcal{A}_1|\mathcal{A}_2}_{2\text{-loc}}(\rho_{_N}) \geq 0 .
\label{eq:positivity_bi}
\end{eqnarray}

\textbf{(ii) Vanishes for product states:} {\color{black} If the pure state $\rho_{_N}$ is bi-separable (product) across the $\mathcal{A}_1|\mathcal{A}_2$ partition, then the reduced state of subsystem $\mathcal{A}_1$ is pure, implying $\nu_i^{\mathcal{A}_1}=1$ for all $i$. Consequently, using Eq.~(\ref{eq:3}), we obtain $\delta^{\mathcal{A}_1|\mathcal{A}_2}_{2\text{-loc}}(\rho_{_N})=0$.
 }

\textbf{(iii) Monotonicity:}
\(N\)-modes pure Gaussian state bipartitioned into \(\mathcal{A}_1\) and \(\mathcal{A}_2\), containing \(m_1\) and \(m_2\) modes 
can be brought--via local Gaussian unitaries--into a standard form consisting of a tensor product of \(m_1\) two-mode squeezed states and \((m_2-m_1)\) vacuum modes \cite{Botero03}.
 Therefore, \(N\)-modes Gaussian pure state $\ket{\psi}_{_N}$, with density matrix $\rho_{_N}=\ket{\psi}_{_N}\bra{\psi}_{_N}$, admits the decomposition as 
\begin{eqnarray}
\ket{\psi}_{_N}=
\bigotimes_{i=1}^{m_1} \ket{\psi}_{i,i}
~\otimes~
\bigotimes_{j=m_1+1}^{m_2} \ket{0}_{j},
\label{eq:any_pure_state}
\end{eqnarray}
where \(\ket{\psi}_{i,i}\) are two-mode squeezed vacuum states characterized by squeezing parameters $r_i$ and CM as,
\begin{eqnarray}
S_{i,i}(r_i)
&=&
\begin{pmatrix}
\cosh 2r_i & 0            & \sinh 2r_i & 0 \\
0          & \cosh 2r_i   & 0           & -\sinh 2r_i \\
\sinh 2r_i & 0            & \cosh 2r_i  & 0 \\
0          & -\sinh 2r_i  & 0           & \cosh 2r_i
\end{pmatrix}.
\label{eq:squ_mat}
\end{eqnarray}
with $r_i \geq 0$  and the symplectic eigenvalue of the CM of \(i\)-th mode obtain from Eq.~(\ref{eq:squ_mat}), is \(\nu_{i}^{\mathcal{A}_1} = \cosh(2r_i) \), which is convex and monotonically increasing function for \(r_i\geq 0\). Similarly, consider another pure \(N\)-modes Gaussian states with respect to a bipartition $\mathcal A_1'$ and $\mathcal A_2'$, containing $m_1$ and $m_2$ modes,
can be written in the standard form
\begin{eqnarray}
\ket{\psi'}_{_N}=
\bigotimes_{i=1}^{m_1}\ket{\psi'}_{i', i'}
\otimes
\bigotimes_{j=m_1+1}^{m_2}\ket{0}_{j'},
\end{eqnarray}
where $\ket{\psi'}_{i', i'}$ denotes a two-mode squeezed vacuum state same form as Eq.~(\ref{eq:squ_mat}),
with the squeezing parameter $r_i'$.


{ Also, it is known that a pure bipartite Gaussian state $\ket{\psi}_{_N}$ can be transformed into another pure bipartite Gaussian state $\ket{\psi'}_{_N}$ by Gaussian local operations and classical communication (GLOCC) if and only if the corresponding vectors of squeezing parameters, $\vec r=(r_1,\ldots,r_{m_1})$ and $\vec r\,'=(r'_1,\ldots,r'_{m_1})$ satisfy the weak majorization relation $\vec r \succ \vec r\,'$ \cite{Giedke03}. Now, if $\vec r$ weakly majorizes $\vec r\,'$, then for any convex and increasing function $f$, one has $\sum_{i=1}^{m_1} f(r_i)\ge \sum_{i=1}^{m_1} f(r'_i)$. In the domain $r_i\ge0$, the function $\nu_i^{\mathcal A_1}=\cosh(2r_i)$ is convex and monotonically increasing in $r_i$. Therefore,
\begin{eqnarray}
\sum_{i=1}^{m_1}\cosh(2r_i) \ge \sum_{i=1}^{m_1}\cosh(2r'_i).
\label{eq:karamata_ineq}
\end{eqnarray}
 }

Equivalently, in terms of symplectic eigenvalues, this yields
\begin{eqnarray}
\sum_{i=1}^{m_1} \nu_{i}^{{\mathcal{A}_1}} ~ \geq ~ \sum_{i=1}^{m_1} {\nu_{i}}^{{\mathcal{A}'_1}} ,
\end{eqnarray}
where
$\{\nu_{i}^{{\mathcal{A}_1}}\}_{i=1}^{m_1}$ and $\{{\nu_{i}}^{{\mathcal{A}'_1}}\}_{i=1}^{m_1}$ denote the symplectic eigenvalues associated with $\ket{\psi}_{_N}$ and $\ket{\psi'}_{_N}$, respectively. Thus, for the \(2\)-local ergotropic gap for a pure bipartite Gaussian state in Eq.~(\ref{eq:gap_pure}), it directly follows that
\begin{eqnarray}
\delta^{\mathcal{A}'_1|\mathcal{A}'_2}_{2\text{-loc}}(\rho_{_N}')&=&
\sum_{i=1}^{m_1} ({\nu_i}^{{\mathcal{A'}_1}} - 1)
\;\leq\;
\sum_{i=1}^{m_1} (\nu_i^{{\mathcal{A}_1}} - 1)
=
\delta^{\mathcal{A}_1|\mathcal{A}_2}_{2\text{-loc}}(\rho_{_N}).\nonumber\\
\label{eq:monotonicity_proof}
\end{eqnarray}

Therefore, \(2\)-local ergotropic gap, \(\delta^{\mathcal{A}_1|\mathcal{A}_2}_{2\text{-loc}}(\rho_{_N})\) is monotonic under GLOCC. Hence, it is a valid measure of bipartite pure state entanglement in the continuous variable Gaussian framework.

\subsection{Proof of Corollary~\ref{co:corollary2}}
\label{app:theorem_1_additional}
\textbf{(i) Faithfulness:}
If the state \(\ket{\psi}_{_N}\) is separable across the bipartition \(\mathcal{A}_1|\mathcal{A}_2\), with \(m_1\) and \(m_2\) modes respectively, then it can be written as \(\ket{\psi}_{_N} = \ket{\psi}_{m_1} \otimes \ket{\psi}_{m_2}\). Consequently, the reduced states corresponding to this bipartition are pure, implying that all symplectic eigenvalues of their covariance matrices are unity, i.e., \(\nu_i^{\mathcal{A}_1} = 1,\ \forall i\). Hence, \(\delta^{\mathcal{A}_1|\mathcal{A}_2}_{2\text{-loc}}(\rho_{_N}) = 0.\)

Conversely, if the \(2\)-local ergotropic gap vanishes, then \(\nu_i^{\mathcal{A}_1} - 1 = 0,\ \forall i\). This implies that the reduced states are pure, and therefore the global pure state \(\ket{\psi}_{_N}\) is separable across the \(\mathcal{A}_1:\mathcal{A}_2\) bipartition.

\textbf{(ii) Strong monotonicity:}
We consider a GLOCC protocol that maps a pure bipartite Gaussian state $\ket{\psi}_{_N}$ onto an ensemble
$\{p_k,\ket{\psi^{(k)}}_{_N}\}$.
For each outcome $k$, the input state, $\ket{\psi}_{_N}$ is transformed into $\ket{\psi^{(k)}}_{_N}$ by any deterministic GLOCC~\cite{Fiur02}.
Since the \(2\)-local ergotropic gap is nonincreasing under deterministic GLOCC, it follows that,
\begin{eqnarray}
\delta^{\mathcal{A}_1|\mathcal{A}_2}_{2\text{-loc}}(\rho_{_N})
\ge \delta^{\mathcal{A}_1|\mathcal{A}_2}_{2\text{-loc}}(\rho_{_N}^{k}),
\qquad \forall\, k .
\label{eq:st_mon_prob}
\end{eqnarray}
Multiplying by the probabilities $p_k$ and summing over all outcomes yields
\begin{eqnarray}
\sum_k p_k \,
\delta^{\mathcal{A}_1|\mathcal{A}_2}_{2\text{-loc}}(\rho_{_N}^{k})
\le
\delta^{\mathcal{A}_1|\mathcal{A}_2}_{2\text{-loc}}(\rho_{_N}),
\label{eq:strong_all_monotone}
\end{eqnarray}
which establishes the strong monotonicity of the \(2\)-local ergotropic gap under GLOCC for pure \(N\)-mode Gaussian states in partitioned \(\mathcal{A}_1|\mathcal{A}_2\).

\textbf{(iii) Additivity:}
Let us Consider two independent pure \(N(N')\)-mode Gaussian states,
$\ket{\psi}_{_N}(\ket{\psi}_{_{N'}})$ and
  split into two partition  of  containing $m_1(m_1')$ and \(m_2(m_2')\) modes with \(m_1\leq m_2(m_1'\leq m_2')\), which label by \(\mathcal{A}_1(\mathcal{A}_1')\) and \(\mathcal{A}_2(\mathcal{A}_2')\) . To calculate  the Gaussian ergotropic gap of composit system, first we group the hole system into two partition which are \(\tilde{\mathcal{A}}_1=\mathcal{A}_1\mathcal{A}_1'\) and \(\tilde{\mathcal{A}}_2=\mathcal{A}_2\mathcal{A}_2'\)  and satisfies,
\begin{eqnarray}
\delta^{\tilde{\mathcal{A}}_1|\tilde{\mathcal{A}}_2}_{2\text{-loc}}(\rho_{_N}\otimes\rho_{_{N'}})&=&\delta^{\mathcal{A}_1|\mathcal{A}_2}_{2\text{-loc}}(\rho_{_N})
+\delta^{\mathcal{A}_1'|\mathcal{A}_2'}_{2\text{-loc}}(\rho_{_{N'}}).\nonumber\\
\label{eq:additivity_biparty}
\end{eqnarray}
To prove this,
the state of subsystem $\tilde{\mathcal{A}}_1$ therefore has
$m_1+m_1'$ symplectic eigenvalues, given by 
$\{\{\nu_i^{A_1}\}_{i=1}^{m_1} \cup \{\nu_i^{A_1'}\}_{i=1}^{m_1'}\}$. Therefore, from Eq.~(\ref{eq:gap_pure}), we obtain
\begin{eqnarray}
\nonumber \delta^{\tilde{\mathcal{A}}_1|\tilde{\mathcal{A}}_2}_{2\text{-loc}}(\rho_{_N}\otimes\rho_{_{N'}}) &=&
\sum_{i=1}^{m_1} (\nu_i^{\mathcal{A}_1} - 1)
+ \sum_{i=1}^{m_1'} (\nu_i^{\mathcal{A}_1'} - 1),\\
&=&
\delta^{\mathcal{A}_1|\mathcal{A}_2}_{2\text{-loc}}(\rho_{_N})
+\delta^{\mathcal{A}_1'|\mathcal{A}_2'}_{2\text{-loc}}(\rho_{_{N'}}).
\end{eqnarray}
Hence, additivity is proof.

\section{Proof of Theorem \ref{th:renyi-2upperbound}}
\label{app:th2}
Since $\rho_{_N}$ is pure, the Rényi-$2$ entanglement entropy equals the Rényi-$2$ entropy
of the reduced state in the $\mathcal{A}_1$, which contains \(m_1\) modes, denoted by $\rho_{m_1}$,
\begin{equation}
S_2(\mathcal{A}_1{:}\mathcal{A}_2)
= S_2(\rho_{_{m_1}}) = -\log \text{Tr}(\rho_{m_1}^2)
= \frac{1}{2}\log \det \rho_{_{m_1}}
= \sum_{j=1}^{m_1} \log \nu_j .
\label{eq:renyi2ee}
\end{equation}

For a pure Gaussian state, the $2$-local ergotropic gap across the bipartition
$\mathcal{A}_1|\mathcal{A}_2$ from Eq. \eqref{eq:gap_pure} is given by
\begin{eqnarray}
\delta_{2\text{-}{\tt loc}}^{\mathcal{A}_1|\mathcal{A}_2}(\rho_{_N})
= \sum_{j=1}^{m_1} (\nu_j - 1),
\label{eq:any_two_party_ergotropy_gap}
\end{eqnarray}
where each $\nu_j \ge 1$, with equality if and only if the corresponding mode
is unentangled.

Using the elementary inequality $\log x \le x - 1$ for all $x \ge 1$, with equality
iff $x=1$, we obtain
\begin{equation}
\sum_{j=1}^{m_1} \log \nu_j
\;\le\;
\sum_{j=1}^{m_1} (\nu_j - 1)
= \delta_{2\text{-}{\tt loc}}^{\mathcal{A}_1|\mathcal{A}_2}(\rho_{_N}).
\end{equation}

Hence,
\begin{equation}
S_2(\mathcal{A}_1{:}\mathcal{A}_2)
\;\le\;
\delta_{2\text{-}{\tt loc}}^{\mathcal{A}_1|\mathcal{A}_2}(\rho_{_N}),
\end{equation}
with equality if and only if $\nu_j = 1$ for all $j$, i.e., when 
$\rho_{_N}$ is separable (product) across $\mathcal{A}_1|\mathcal{A}_2$.

The second half of the proof establishes the impossibility of any functional independence between $S_2(\mathcal{A}_1{:}\mathcal{A}_2)$ and $\delta_{2\text{-}{\tt loc}}^{\mathcal{A}_1|\mathcal{A}_2}(\rho_{_N})$. We prove this by contradiction. We first assume there exists a function $F:\mathbb{R}_{\ge 0}\to\mathbb{R}_{\ge 0}$ such that
\begin{eqnarray}
    \delta^{\mathcal{A}_1|\mathcal{A}_2}_{2-\mathrm{loc}}(\rho_{_N})
= F\!\left(S_2(\mathcal{A}_1:\mathcal{A}_2)\right)
\label{eq:functionaldep}
\end{eqnarray}
for all pure Gaussian states $\rho_{_N}$. Then any two
pure Gaussian states with the same Rényi-$2$ entanglement entropy across the bipartition $\mathcal{A}_1|\mathcal{A}_2$ must necessarily
have the same $2$-local ergotropic gap. For a pure Gaussian state, both quantities depend on the symplectic spectrum
$\boldsymbol{\nu}=(\nu_1,\ldots,\nu_{m_1})$ of $\rho_{_{m_1}}$ following Eqs. \eqref{eq:gap_pure} and \eqref{eq:renyi2ee}. Now we can re-express
\begin{equation}
S_2(\rho_{_{m_1}}) = \sum_{j=1}^{m_1} \log \nu_j
= \log\!\Big(\prod_{j=1}^{m_1} \nu_j\Big),
\end{equation}

We now focus on the minimal nontrivial case $m_1=2$. Fix a constant $C>1$ and consider all spectra
$(\nu_1,\nu_2)$ such that
\begin{equation}
\nu_1 \nu_2 = C .
\label{eq:constraint}
\end{equation}
For all such pairs one has
\begin{equation}
S_2(\rho_{_{m_1}}) = \log \nu_1 + \log \nu_2 = \log C ,
\end{equation}
so the Rényi-$2$ entanglement entropy is the same for all states satisfying \eqref{eq:constraint}. However, the
corresponding 2-local ergotropic gap reads
\begin{equation}
\delta_{2\text{-}{\tt loc}}^{\mathcal{A}_1|\mathcal{A}_2}
= (\nu_1 - 1) + (\nu_2 - 1)
= \nu_1 + \nu_2 - 2 .
\label{eq:gap2}
\end{equation}
Under the constraint \eqref{eq:constraint}, the sum $\nu_1 + \nu_2$ is not fixed: for instance, the symmetric
choice $\nu_1=\nu_2=\sqrt{C}$ and the asymmetric choice $\nu_1=C,\;\nu_2=1$ both satisfy
\eqref{eq:constraint}, but give
\begin{equation}
\delta_{2\text{-}{\tt loc}}^{\mathcal{A}_1|\mathcal{A}_2,\mathrm{sym}} = 2(\sqrt{C}-1),
\qquad
\delta_{2\text{-}{\tt loc}}^{\mathcal{A}_1|\mathcal{A}_2,\mathrm{asym}} = C-1 .
\end{equation}
For any $C>1$ these two values are different. Hence we have exhibited two pure Gaussian states with the same
Rényi-$2$ entanglement entropy but different 2-local ergotropic gaps, thereby arriving at a contradiction to the assumption in Eq. \eqref{eq:functionaldep}.  Therefore, for $m_1 \ge 2$,
$S_2(\rho_{m_1})$ and $\delta_{2\text{-}{\tt loc}}^{\mathcal{A}_1|\mathcal{A}_2}(\rho_{_N})$ are functionally independent. Hence the proof.

Moreover, it is interesting to ask whether there are special cases in which the $2$-local ergotropic gap and the Rényi-$2$ entanglement entropy do share a functional relationship. 

\begin{corollary}[Functional dependence in special regimes]
 For   a pure Gaussian state $\rho_{_N}$ on $N=m_1+m_2$ modes, bipartitioned as $\mathcal{A}_1|\mathcal{A}_2$ with $m_1 \le m_2$, the $2$-local ergotropic gap
 $\delta_{2\text{-}{\tt loc}}^{\mathcal{A}_1|\mathcal{A}_2}(\rho_{_N})$ and the Rényi-$2$ entanglement entropy $S_2(\mathcal{A}_1{:}\mathcal{A}_2)$ are (exactly or approximately)
 functionally dependent in the following regimes:
 \begin{enumerate}
 \item \textbf{Product states:} $\rho_{_N} = \rho_{m_1}\otimes \rho_{m_2}$;
 \item \textbf{Two–mode case:} $m_1 = m_2 = 1$;
 \item \textbf{Flat entanglement spectrum:} $\nu_1=\cdots=\nu_{m_1}\equiv\nu$;
 \item \textbf{Nearly flat spectrum:} $\nu_j=\nu+\delta_j$ with $\sum_j \delta_j=0$ and $|\delta_j|\ll \nu$, ~\emph{[approximate dependence]}; and 
\item \textbf{Weak entanglement regime:} $\nu_j=1+\epsilon_j$ with \\ $0<\epsilon_j\ll1$, ~\emph{[approximate dependence].}
 \end{enumerate}
 \label{cor:funcdep1}
 \end{corollary}

\begin{proof}
For a pure Gaussian state, both quantities depend only on the symplectic
spectrum $\{\nu_j\}$ of the reduced state on $\mathcal{A}_1$:
\[
\delta_{2\text{-}{\tt loc}}^{\mathcal{A}_1|\mathcal{A}_2}(\rho_{_N})
= \sum_{j=1}^{m_1}(\nu_j-1),
\qquad
S(\mathcal{A}_1{:}\mathcal{A}_2)
= \sum_{j=1}^{m_1} \log \nu_j,
\]

\medskip
\noindent\textbf{(1) Product state.}
If $\rho_{_N}=\ket{\psi}_{m_1}\otimes \ket{\psi}_{m_2}$, then
$\nu_j=1$ for all $j$. Hence
\[
\delta_{2\text{-}{\tt loc}}^{\mathcal{A}_1|\mathcal{A}_2}=0,
\qquad
S(\mathcal{A}_1{:}\mathcal{A}_2)=0.
\]
Both quantities vanish identically, so they are trivially functionally related.

\medskip
\noindent\textbf{(2) Two–mode case $m_1=m_2=1$.}
There is a single symplectic eigenvalue $\nu$. Thus
\[
\delta_{2\text{-}{\tt loc}}^{\mathcal{A}_1|\mathcal{A}_2}=\nu-1,
\qquad
S(\mathcal{A}_1{:}\mathcal{A}_2)=\log \nu.
\]
Eliminating $\nu$ gives
\[
S(\mathcal{A}_1{:}\mathcal{A}_2)
= \log\!\bigl(1+\delta_{2\text{-}{\tt loc}}^{\mathcal{A}_1|\mathcal{A}_2}\bigr),
\]
which is an exact functional dependence.

\medskip
\noindent\textbf{(3) Flat spectrum $\nu_j=\nu$.}
Then
\[
\delta_{2\text{-}{\tt loc}}^{\mathcal{A}_1|\mathcal{A}_2}
= m_1(\nu-1),
\qquad
S(\mathcal{A}_1{:}\mathcal{A}_2)=m_1 \log \nu.
\]
Solving for $\nu$ gives
\[
\nu = 1 + \frac{1}{m_1}\,
\delta_{2\text{-}{\tt loc}}^{\mathcal{A}_1|\mathcal{A}_2},
\]
and therefore
\[
S(\mathcal{A}_1{:}\mathcal{A}_2)
= m_1\, \log\!\left(1+\frac{1}{m_1}
\delta_{2\text{-}{\tt loc}}^{\mathcal{A}_1|\mathcal{A}_2}\right),
\]
which is again an exact functional relation.

\medskip
\noindent\textbf{(4) Nearly flat spectrum.}
Let $\nu_j=\nu+\delta_j$ with $\sum_j\delta_j=0$ and $|\delta_j|\ll\nu$. Then
\[
\delta_{2\text{-}{\tt loc}}^{\mathcal{A}_1|\mathcal{A}_2}
= \sum_j(\nu_j-1)=m_1(\nu-1),
\]
so $\nu$ is fixed by $\delta_{2\text{-}{\tt loc}}$ as before. For the Rényi-$2$ entanglement entropy, expand $\log (\nu + \delta_j)$:
\[
\log(\nu+\delta_j)
= \log \nu + \frac1\nu\delta_j - \tfrac{1}{2\nu^2} \delta_j^2 + O(\delta_j^3).
\]
Summing over $j$ and using $\sum_j\delta_j=0$ gives
\[
S(\mathcal{A}_1{:}\mathcal{A}_2)
= m_1 \log \nu - \tfrac{1}{2\nu^2} \sum_j \delta_j^2 + O(\delta_j^3).
\]
Thus
\[
S(\mathcal{A}_1{:}\mathcal{A}_2)
= m_1 \log\!\left(1+\frac{1}{m_1}
\delta_{2\text{-}{\tt loc}}^{\mathcal{A}_1|\mathcal{A}_2}\right)
-\frac{1}{2\nu^2} O\big(\sum_j\delta_j^2\big),
\]
showing an approximate single-variable dependence with corrections quadratic
in the spectral deviations.

\medskip
\noindent\textbf{(5) Weak entanglement regime.}
Let $\nu_j = 1 + \epsilon_j$ with $0<\epsilon_j\ll1$. Then
\[
\delta_{2\text{-}{\tt loc}}^{\mathcal{A}_1|\mathcal{A}_2}
= \sum_{j=1}^{m_1} \epsilon_j.
\]
We expand $\log \nu$ in a Taylor series around $\nu = 1$. Since $\log 1=0$ and
$f'(1)=1$, we have:
\[
\log(1+\epsilon)
=  \epsilon  - O(\epsilon^2).
\]
Therefore the Rényi-$2$ entanglement entropy becomes
\[
S(\mathcal{A}_1{:}\mathcal{A}_2)
=  \sum_{j=1}^{m_1} (\epsilon_j - O(\epsilon_j^2)).
\]
Therefore we have 
\[
S(\mathcal{A}_1{:}\mathcal{A}_2)
= \delta_{2\text{-}{\tt loc}}^{\mathcal{A}_1|\mathcal{A}_2} - \sum_jO(\epsilon_j^2),
\]
showing that in the weak-entanglement regime the two quantities are
approximately equal with corrections quadratic
in the spectral deviations.    
\end{proof}

\section{Connection between $2$-local ergotropy and von-Neumann entropy}
\label{app:ee}
Here, we consider the case of analyzing von Neumann entropy as a measure of entanglement entropy and its connection to the $2$-local ergotropic gap. First, we will show that the $2$-local ergotropy is functionally independent from the von-Neumann entropy of the marginal system and, equivalently, the mutual information for pure states. Next, we aim to show that the $2$-local ergotropic gap bounds the mutual information.

\subsection{\(2\)-local ergotropic gap is functionally independent from mutual information}
\label{sec:mi_eg_not}
\begin{proposition}
    For pure Gaussian states with $m_1+m_2$ modes, the $2$-local ergotropic gap
$\delta^{\mathcal{A}_1|\mathcal{A}_2}_{2\text{-loc}}(\rho_{_N})$ is, in general, not a function of the mutual information between the two partitions.
\label{pro:pro1}
\end{proposition}
 \begin{proof}
Consider a pure \(N\)-mode Gaussian state $\rho_{_N}$, partitioned in \(\mathcal{A}_1|\mathcal{A}_2\) having \(m_1\) and \(m_2\) number of modes. Therefore, the mutual information takes the form as
\begin{eqnarray}
\mathrm{MI}(\rho_{_N})
&=&
S(\rho_{_{m_1}})+S(\rho_{_{m_2}})-S(\rho_{_N})
\nonumber\\
&=&
2S(\rho_{_{m_1}}),
\end{eqnarray}
 where \(\rho_{_{m_i}}=\Tr_{{\mathcal{A}_i}}\rho_{_N}\) and $S(\cdot)$ denotes the von Neumann entropy.
 In terms of the symplectic eigenvalues $\{\nu_i^{\mathcal{A}_1}\}_{i=1}^{m_1}$ of the reduced CM of subsystem $\mathcal{A}_1$, the mutual information can be written as
\begin{eqnarray}
\mathrm{MI}(\rho_{_N})= 2\sum_{i=1}^{m_1}
\big[\frac{\nu_i^{\mathcal{A}_1}+1}{2}\ln\frac{\nu_i^{\mathcal{A}_1}+1}{2}
- \frac{\nu_i^{\mathcal{A}_1}-1}{2}\ln\frac{\nu_i^{\mathcal{A}_1}-1}{2}\big].
\label{eq:mutual_inf_sym}
\end{eqnarray}
 Both the ergotropic gap $\delta^{\mathcal{A}_1|\mathcal{A}_2}_{2\text{-loc}}(\rho_{_N})$ and the mutual information
 $\mathrm{MI}(\rho_{_N})$ are functions of the symplectic eigenvalues $\nu_i^{\mathcal{A}_1}$.
 The Jacobian matrix associated with the mapping
 $\{\delta^{\mathcal{A}_1|\mathcal{A}_2}_{2\text{-loc}}(\rho_{_N}),\mathrm{MI}(\rho_{_N})\}$ is given by
\begin{eqnarray}
\nonumber J &=&
\begin{pmatrix}
\frac{\partial \delta^{\mathcal{A}_1|\mathcal{A}_2}_{2\text{-loc}}(\rho_{_N})}{\partial \nu_1^{\mathcal{A}_1}} &
\frac{\partial \delta^{\mathcal{A}_1|\mathcal{A}_2}_{2\text{-loc}}(\rho_{_N})}{\partial \nu_2^{\mathcal{A}_1}} &
\cdots &
\frac{\partial \delta^{\mathcal{A}_1|\mathcal{A}_2}_{2\text{-loc}}(\rho_{_N})}{\partial \nu_{m_1}^{\mathcal{A}_1}}
\\
\frac{\partial \mathrm{MI}(\rho_{_N})}{\partial \nu_1^{\mathcal{A}_1}} &
\frac{\partial \mathrm{MI}(\rho_{_N})}{\partial \nu_2^{\mathcal{A}_1}} &
\cdots &
\frac{\partial \mathrm{MI}(\rho_{_N})}{\partial \nu_{m_1}^{\mathcal{A}_1}}
\end{pmatrix},\\
&=&
\begin{pmatrix}
1 & 1 & \cdots & 1 \\
\ln\frac{\nu_1^{\mathcal{A}_1}+1}{\nu_1^{\mathcal{A}_1}-1} &
\ln\frac{\nu_2^{\mathcal{A}_1}+1}{\nu_2^{\mathcal{A}_1}-1} &
\cdots &
\ln\frac{\nu_{m_1}^{\mathcal{A}_1}+1}{\nu_{m_1}^{\mathcal{A}_1}-1}
\end{pmatrix}.
\end{eqnarray}

 From this Jacobian, we identify the two vectors
\begin{eqnarray}
\ket{v_1} &=& (1,1,\ldots,1)^{T}, \nonumber\\
\ket{v_2} &=&
\left(
\ln\frac{\nu_1^{\mathcal{A}_1}+1}{\nu_1^{\mathcal{A}_1}-1},
\ln\frac{\nu_2^{\mathcal{A}_1}+1}{\nu_2^{\mathcal{A}_1}-1},
\ldots,
\ln\frac{\nu_{m_1}^{\mathcal{A}_1}+1}{\nu_{m_1}^{\mathcal{A}_1}-1}
\right)^{T}.
\label{eq:vector}
\end{eqnarray}
These vectors are linearly dependent if and only if
\begin{eqnarray}
\frac{\nu_i^{\mathcal{A}_1}+1}{\nu_i^{\mathcal{A}_1}-1} = c,
~ \forall\, i,
\end{eqnarray}
for some constant $c$, which implies $\nu_i^{\mathcal{A}_1}=\nu_j^{\mathcal{A}_1}$ for all $i,j$.
In general, the symplectic eigenvalues need not be equal for $m_1\ge2$.
Therefore, the \(2\)-local ergotropic gap and the mutual information are, in general, functionally independent for generic bipartite Gaussian states with $m_1\ge2$. However, in the special case of two-mode pure Gaussian states
($m_1=m_2=1$), the $2$-local ergotropic gap
$\delta^{\mathcal{A}_1|\mathcal{A}_2}_{2\text{-loc}}(\rho_{_2})$ becomes a function of the
von Neumann entropy, as discussed in Ref.~\cite{Polo2025}. Hence, the proof.
 \end{proof}

 \subsection{$2$-local ergotropic gap bounds mutual information}

\begin{proposition}[$2$-local ergotropic gap bounds mutual information]
Let $\rho_{_N}$ be a pure bipartite Gaussian state with respect to the partition
$\mathcal{A}_1|\mathcal{A}_2$. Then the $2$-local ergotropic gap satisfies
\begin{equation}
\frac{1}{2}\, \emph{MI}(\rho_{_N}) < \delta_{2\text{-}{\tt loc}}^{\mathcal{A}_1|\mathcal{A}_2}(\rho_{_N}) + \min \{ |\mathcal{A}_1|,|\mathcal{A}_2|\}
\end{equation}
where $\emph{MI}(\rho_{_N})$ is the quantum mutual information, and $|\mathcal{A}_j|$ denotes the number of modes in the partition $\mathcal{A}_j.$
\end{proposition}

\begin{proof}
Since $\rho_{_N}$ is pure, the mutual information reduces to
\[
\text{MI}(\rho_{_N})
= 2\,S(\rho_{m_1}) = -2 \text{Tr}(\rho_{m_1}\log \rho_{m_1}).
\]
Here we assume the number of modes in partitions $\mathcal{A}_1$ and $\mathcal{A}_2$ are given by $m_1$ and $m_2$ respectively $(m_1+m_2 = N)$. Without any loss of generality, consider $m_1\leq m_2$.
For an $m_1$-mode Gaussian state with symplectic spectrum
$\{\nu_j\}_{j=1}^{m_1}$, the von Neumann entropy is given by \cite{GRMP}
\begin{eqnarray}
    S(\rho_{m_1})
= \sum_{j=1}^{m_1} f(\nu_j),
\end{eqnarray}
where $f(\nu)
= \frac{\nu+1}{2}\log\!\frac{\nu+1}{2}
- \frac{\nu-1}{2}\log\!\frac{\nu-1}{2}.$ Hence,
\[
\text{MI}(\rho_{_N})
= 2 \sum_{j=1}^{m_1} f(\nu_j).
\]

On the other hand, for a pure Gaussian state the $2$-local ergotropic gap is
\[
\delta_{2\text{-}{\tt loc}}^{\mathcal{A}_1|\mathcal{A}_2}(\rho_{_N})
= \sum_{j=1}^{m_1} (\nu_j - 1).
\]
We know that for all $\nu \ge 1$, we have the following bound:
\[
f(\nu) < \log \nu + 1 < \nu.
\]
Applying this bound term by term gives
\[
\text{MI}(\rho_{_N})
= 2 \sum_{j=1}^{m_1} f(\nu_j)
\;<\;
2 \sum_{j=1}^{m_1} (\nu_j - 1) +2m_1
,
\]
which is equivalent to
\begin{equation}
S(\rho_{m_1})=\frac{1}{2}\, \text{MI}(\rho_{_N}) < \delta_{2\text{-}{\tt loc}}^{\mathcal{A}_1|\mathcal{A}_2}(\rho_{_N}) + |\mathcal{A}_1|.
\end{equation}
If we have assumed $m_2\leq m_1$, we would have obtained
\begin{equation}
S(\rho_{m_1})=\frac{1}{2}\, \text{MI}(\rho_{_N}) < \delta_{2\text{-}{\tt loc}}^{\mathcal{A}_1|\mathcal{A}_2}(\rho_{_N}) + |\mathcal{A}_2|.
\end{equation}
This proves the claim.
\end{proof}

\section{Proof of Theorem \ref{th:th4}}
\label{app:min_ergo_k_party}
To establish that $\Delta^{k}(\rho_{_N})$ is a valid entanglement measure, we verify the
required properties.

\textbf{(i) Positivity:}
For any bipartition, the \(2\)-local ergotropic gap is nonnegative, i.e.,
\( \delta ^{\mathcal{A}_j|\bar{\mathcal{A}_j}}_{_{2-\tt loc}}(\rho_{_N})\ge 0 \) (see Eq.~\ref{eq:positivity_bi}), where $\bar {\mathcal{A}_j}$ denotes the subsystem containing all modes except those in $\mathcal{A}_j$ and
\(
\delta ^{\mathcal{A}_j|\bar{\mathcal{A}_j}}_{_{2-\tt loc}}(\rho_{_N})=\sum_{i=1}^{m_j}\big(\nu_i^{\mathcal{A}_j} - 1\big),
\)
where $\{\nu_i^{\mathcal{A}_j}\}_{i=1}^{m_j}$ denote the symplectic eigenvalues of
the reduced covariance matrix of partition $\mathcal{A}_j$. Since \(k\)-local ergotropic gap, $\delta^{\mathcal{A}_1|\mathcal{A}_2|\cdots|\mathcal{A}_k}_{k-\tt loc}
(\rho_{_N})$ can be computed in terms of sum of \(2\)-local ergotropic gap as
\begin{eqnarray}
\delta^{\mathcal{A}_1|\mathcal{A}_2|\cdots|\mathcal{A}_k}_{k-\tt loc}
(\rho_{_N})&=&\frac{1}{2}\sum_{j=1}^{k}\delta ^{\mathcal{A}_j|\bar{\mathcal{A}_j}}_{_{2-\tt loc}}(\rho_{_N}),\nonumber\\&=&\frac12\sum_{j=1}^{k}\sum_{i=1}^{m_j}(\nu_{i}^{\mathcal{A}_j}-1).
    \label{eq:k_loc_2_loc}
\end{eqnarray}
where each individual term, \(\delta ^{\mathcal{A}_j|\bar{\mathcal{A}_j}}_{_{2-\tt loc}}(\rho_{_N})\geq0\) indicates that \(\delta^{\mathcal{A}_1|\mathcal{A}_2|\cdots|\mathcal{A}_k}_{k-\tt loc}
(\rho_{_N})\geq 0\) and therefore \(\),
it follows that the minimization over all possible different \(k\)-partitions is also positive, i.e., 
\begin{eqnarray}
\Delta^{k}(\rho_{_N})&=&\min_{\mathcal{A}_1|\mathcal{A}_2|\ldots|\mathcal{A}_k}\delta ^{\mathcal{A}_1|\mathcal{A}_2|\ldots|\mathcal{A}_k}_{k-\tt loc}(\rho_{_N})\geq 0.
\end{eqnarray}
Thus, positivity holds.

{\color{black} \textbf{(ii) Vanishes for $k$-separable (product) states:} If \(\rho_{_N}\) is a \(k\)-separable pure state with respect to a particular \(\mathcal{A}_1|\mathcal{A}_2|\ldots|\mathcal{A}_k\) partition, then the reduced states of all the subsystems $\{\mathcal{A}_i\}_{i=1}^k$ are pure state which implies \(\nu_i^{\mathcal{A}_j}=1 ~\forall~ i,j\). Therefore, from Eq.~(\ref{eq:k_loc_2_loc}) we get $\delta^{\mathcal{A}_1|\mathcal{A}_2|\cdots|\mathcal{A}_k}_{k-\tt loc}=0$ and consequently it follows that \(\Delta^k(\rho_{_N})=0\). }

\textbf{(iii) Monotonicity:}
Since the \(2\)-local ergotropic gap,  is monotonic under GLOCC. So, \(k\)-local ergotropic gap will also be  GLOCC monotone as it is the sum of all \(k\)-possible \(2\)-local ergotropic gaps as shown in Eq.~(\ref{eq:k_loc_2_loc}), therefore,
\(\Delta^{(k)}(\rho_{_N})\) is also monotone under GLOCC as it minimizes the \(k\)-local ergotropic gap over the all possible \(k\)-partitions.

\subsection{Proof of Corollary \ref{co:corollary4}}
\label{app:prove_corollary_4}

\textbf{(i) Faithfulness:} Let us consider \(\rho_{_N}\) be a pure \(N\)-mode Gaussian state such that
\(\Delta^{k}(\rho_{_N})=0\) and then there exists a \(k\)-partition
\(\mathcal{A}_1^{*}| \mathcal{A}_2^{*}| \cdots| \mathcal{A}_k^{*}\) where \(i\)-th partition contain \(m_i^*\) modes, for which \(\delta^{\mathcal{A}_1^*|\mathcal{A}_2^*|\cdots|\mathcal{A}_k^*}_{k-\tt loc}
(\rho_{_N})=0\). Therefore, from Eq.~(\ref{eq:k_loc_2_loc}) we get,
\begin{eqnarray}
\frac{1}{2}\sum_{j=1}^{k}
\delta ^{\mathcal{A}_j^*|\bar{\mathcal{A}_j^*}}_{_{2-\tt loc}}(\rho_{_N})&=&0 .
\label{eq:k_loc_zero}
\end{eqnarray}
Therefore, the only solution possible of Eq.~(\ref{eq:k_loc_zero}) is \(\delta ^{\mathcal{A}_j^*|\bar{\mathcal{A}_j^*}}_{_{2-\tt loc}}(\rho_{_N})=0,~\forall j\) as each \(2\)-local ergotropic gap is non-negative. This also enforces that all the simplectic eigenvalues of Eq.~(\ref{eq:k_loc_2_loc}) is equal to unity, i.e., \(\nu_{i}^{\mathcal{A}_j}=1,~\forall i,j\) which implies that each reduced state
\(\rho_{_{m_j^{*}}}\) is pure and the global pure state can be written as, 
\(\rho_{_N} = \bigotimes_{j=1}^{k}\rho_{_{m_j^{*}}},\) and hence the state is \(k\)-separable.

Conversely, consider the state \(\rho_{_N}\) is
\(k\)-separable with respect to a partition
\(\mathcal{A}_1|\mathcal{A}_2| \cdots|\mathcal{A}_k \) and for this scenario, the corresponding \(k\)-local ergotropic gap vanishes,
\(\delta^{\mathcal{A}_1|\mathcal{A}_2|\cdots|\mathcal{A}_k}_{k-\tt loc}
(\rho_{_N})=0 .\) It then follows directly from the definition, \(\Delta^{k}(\rho_{_N})=0\).
This completes the proof.

\textbf{(ii) Additivity:}
Consider two  pure states
\(\rho_{_N}\) and
\(\rho_{_{N'}}\), splitting into \(k\)- partition, containing  \(m_i(m_i')\) modes in partitions \(\mathcal{A}_i(\mathcal{A}_i')\) with \(\sum_{i=1}^{k} m_i =N,~ \text{and} ~\sum_{i=1}^{k} m_i' = N'\). Now, the composit system of total modes, \(N+N'\) split into \(k\)- partition and \(i\)-th partition labeled by $\tilde{\mathcal{A}}_i =\mathcal{A}_i\mathcal{A}_i'$ which contain \(\tilde{m}_i=(m_i+m_i')\) modes. Since $\tilde{\mathcal{A}}_i$ is composed of two separable parts, $\mathcal{A}_i$ and $\mathcal{A}_i'$, the optimal local unitary acting on $\tilde{\mathcal{A}}_i$ necessarily factorizes into independent unitaries on $\mathcal{A}_i$ and $\mathcal{A}_i'$. Consequently, using Eq.~(\ref{eq:k_loc_2_loc}), the \(k\)-local ergotropic gap of the composite state takes the form as 
\begin{eqnarray}
\delta^{\tilde{\mathcal{A}}_1|\tilde{\mathcal{A}}_2|\cdots|\tilde{\mathcal{A}}_k}_{k-\tt loc}
(\rho_{_N}\otimes \rho_{_{N'}})
&=&
\frac{1}{2}\sum_{j=1}^{k}\sum_{i=1}^{\tilde m_j}(\nu_i^{\tilde{\mathcal{A}}_j}-1) \nonumber\\
&=&
\frac{1}{2}\sum_{j=1}^{k}\sum_{i=1}^{m_j}(\nu_i^{\mathcal{A}_j}-1)
+
\frac{1}{2}\sum_{j=1}^{k}\sum_{i=1}^{m_j'}(\nu_i^{\mathcal{A}_j'}-1) \nonumber\\
&=&
\delta^{\mathcal{A}_1|\mathcal{A}_2|\cdots|\mathcal{A}_k}_{k-\tt loc}
(\rho_{_N})
+
\delta^{\mathcal{A}_1'|\mathcal{A}_2'|\cdots|\mathcal{A}_k'}_{k-\tt loc}
(\rho_{_{N'}}),\nonumber\\
\label{eq:add_for_given_k_party}
\end{eqnarray}
where \(\{\nu_i^{\mathcal{A}_j}\}_{i=1}^{m_j}\) and \(\{\nu_i^{\mathcal{A}_j'}\}_{i=1}^{m_j'}\) are the symplectic eigenvalues of the covariance matrices of the subsystems
\(\mathcal{A}_j\) and \(\mathcal{A}_j'\), respectively.

To prove the additivity of \(\Delta^k\), let us consider two identical copies of an \(N\)-mode pure Gaussian state
\(\rho_{_N}\) and we split this two state into \(k\)-patition label by \(\mathcal{A}_i\) and \(\mathcal{A}'_i\) . Now their joint state,
\(\rho_{_{2N}}=\rho_{_N}
\otimes{\rho}_{_N}\),
splitting into \(k\)-partition and each partition label by \(\tilde{\mathcal{A}}_i=\mathcal{A}_i\mathcal{A}_i'\).
The multipartite entanglement of this state is quantified by  \(k\)-ergotropic score where minimization required over all \(k\)-partition and let us consider that this minimum is obtained for an
optimal partition,
\(\tilde{\mathcal{A}}_1^{*}|\tilde{\mathcal{A}}_2^{*}|\cdots|\tilde{\mathcal{A}}_k^{*}\),
where each subset  can be written as \(\tilde{\mathcal{A}}_i^{*}=\mathcal{A}_i^{*}\mathcal{A}_i'^{*}\).
Then using Eq.~(\ref{eq:add_for_given_k_party}), we get
\begin{eqnarray}
\nonumber\Delta^k\big(\rho_{_{2N}})&=&
\delta^{\tilde{\mathcal{A}}_1^*|\tilde{\mathcal{A}}_2^*|\cdots|\tilde{\mathcal{A}}_k^*}_{k-\tt loc}
(\rho_{_{2N}})
\nonumber\\
&=&
\delta^{\mathcal{A}_1^*|\mathcal{A}_2^*|\cdots|\mathcal{A}_k^*}_{k-\tt loc}
(\rho_{_N})
+
\delta^{\mathcal{A}_1'^*|\mathcal{A}_2'^*|\cdots|\mathcal{A}_k'^*}_{k-\tt loc}
(\rho_{_{N}}),\nonumber\\
 &\ge&
\Delta^k\big({\rho}_{_N}\big)
+
\Delta^k\big({\rho}_{_N}\big).
\label{eq:add_k_sepa_1}
\end{eqnarray}
On the other hand, the sum of the multipartite entanglement of the two identical states is given by
\begin{eqnarray}
\nonumber &&\Delta^k(\rho_{_N})+
\Delta^k(\rho_{_N})\\
&=&
\delta^{\mathcal{A}_1^{''}|\mathcal{A}_2^{''}|\cdots|\mathcal{A}_k^{''}}_{k-\text{loc}}\big(\rho_{_N})
+
\delta^{\mathcal{A}_1^{''}|\mathcal{A}_2^{''}|\cdots|\mathcal{A}_k^{''}}_{k-\text{loc}}\big(\rho_{_N})\nonumber \\
&=&
\delta^{\tilde{\mathcal{A}_1^{''}}|\tilde{\mathcal{A}_2^{''}}|\cdots|\tilde{\mathcal{A}_3^{''}}}_{k-\text{loc}}\big(\rho_{_{2N}}),
\end{eqnarray}
where the optimal \(k\)-partition, \(\mathcal{A}_1^{''}|\mathcal{A}_2^{''}|\cdots|\mathcal{A}_k^{''}\) contributes to \(\Delta^{k}(\rho_{_N})\) and \(\tilde{\mathcal{A}_i^{''}}=\mathcal{A}_i^{''}\mathcal{A}_i^{''}\).
However, the important point is that the optimal \(k\)-partition of \(\tilde{\mathcal{A}_1^{''}}|\tilde{\mathcal{A}_2^{''}}|\ldots|\tilde{\mathcal{A}_k^{''}}\) is not necessarily same \(k\)-partition \(\tilde{\mathcal{A}_1^*}|\tilde{\mathcal{A}_2^*}|\ldots|\tilde{\mathcal{A}_k^*}\) from which \(\Delta^k(\rho_{_{2N}})\) is obtained. Moreover, when we are computing \(\Delta^k(\rho_{_N})+\Delta^k(\rho_{_N})\), the optimal \(k\)-partition is different from \(\Delta^k(\rho_{_{2N}})\). Thus we can say, 
\begin{eqnarray}
\delta^{\tilde{\mathcal{A}_1^{''}}|\tilde{\mathcal{A}_2^{''}}|\cdots|\tilde{\mathcal{A}_3^{''}}}_{k-\text{loc}}\big(\rho_{_{2N}})
&\ge&\Delta^k(\rho_{_{2N}}),\nonumber\\
\text{i.e.,}~~\Delta^k(\rho_{_N})+\Delta^k(\rho_{_N})&\ge&\Delta^k(\rho_{_{2N}}).
\label{eq:add_k_sepa_2}
\end{eqnarray}
 Therefore, Eqs.~(\ref{eq:add_k_sepa_1}) and~(\ref{eq:add_k_sepa_2}) are mutually contradictory. Thus the only possible solution is 
 \begin{eqnarray}
\Delta^k(\rho_{_{2N}})&=&\Delta^k(\rho_{_{N}})+\Delta^k(\rho_{_{N}}).
     \label{eq:k_add}
 \end{eqnarray}
 Hence, this completes the proof.

\section{Relation to multimode geometric entanglement measures }
\subsection{Proof of Proposition \ref{pro:relation_ggm_gap}}
\label{App:relation_ggm_gap}
We now focus on the three-mode system ($N=3$), there are three possible bipartitions, denoted by $\mathcal{A}_i|\bar{\mathcal{A}}_i$ with $i=1,2,3$. The corresponding ergotropic gap for each bipartition is given by \(\delta^{\mathcal{A}_i|\bar{\mathcal{A}_i}}_{2-\text{loc}} =(\nu^{\mathcal{A}_i}-1) \). Thus, the $2$-ergotropic score takes the form
\(\Delta^2(\rho_{3}) = \min[\{(\nu^{\mathcal{A}_i}-1)\}_{_{i=1}}^3]\). Similarly, \(G(\rho_{_3}) =1-\max[\big\{\frac{2}{1+\nu^{\mathcal{A}_i}}\big\}_{_{i=1}}^3] \). Since $\nu^{\mathcal{A}_i}$ and $\frac{2}{1+\nu^{\mathcal{A}_i}}$ vary oppositely, the maximum of the latter occurs for the bipartition where
$\nu^{\mathcal{A}_i}$ is minimum. Consequently, both
$\Delta^{2}(\rho_{3})$ and $G(\rho_{3})$ are determined by the same optimal bipartition. Let assume minimum comes from \(j\)-th bipartition, \(\mathcal{A}_j|\bar{\mathcal{A}}_j\). So 
\begin{eqnarray}
    \Delta^2(\rho_{3})&=&\nu^{\mathcal{A}_j}-1 \label{eq:final_3_gap}\\
    G(\rho_{_3}) &=& 1-\frac{2}{1+\nu^{\mathcal{A}_j}}
    \label{eq:final_3_GGM}
\end{eqnarray}
Therefore from  Eq.~(\ref{eq:final_3_gap}) and Eq.~(\ref{eq:final_3_GGM}) we get 
\begin{eqnarray}
     \Delta^2(\rho_{_3})&=& \frac{2~G(\rho_{_3})}{1 - G(\rho_{_3})}.
   \label{app:conn_ggm}
\end{eqnarray}
Here, \(\Delta^{2}(\rho_{_3})\) is in one-to-one correspondence with \(G(\rho_{_3})\), since in both cases, the minimum symplectic eigenvalue is attained for the same bipartition among all possible bipartitions. Moreover, both quantities are monotonic functions of these symplectic eigenvalues as they increase with increasing symplectic eigenvalues and decrease with decreasing symplectic eigenvalues. Hence, the two measures are monotonic and equivalent, which completes the proof.

\subsection{Proof of Proposition. \ref{pro:distance} }
\label{App:distance_base_prove}
For an arbitrary $N$-mode Gaussian quantum state $\rho_{_N}$, the Gaussian total multimode entanglement $\mathcal{E}_N^{\mathcal{G}}$ can be formally expressed as the minimum distance of $\rho_{_N}$ from the set of fully separable Gaussian states.
\begin{eqnarray}
  \mathcal{E}_N^{\mathcal{G}}(\rho_{_N})&=& \underset{{\sigma \in \mathcal{G}_{N-\text{sep}}}}{\min} D(\rho_{_N},\sigma),
  \label{eq:tme-gen}
 \end{eqnarray}
 where $\mathcal{G}_{N-\text{sep}}$ is the set fully separable $N-$mode Gaussian states, and $D(.,.)$ is some distance functional. We can take express distance functional $D$ to be the fidelity distance,  $D(\eta_1,\eta_2) = 1 - \mathcal{F}(\eta_1,\eta_2)$ for arbtrary quantum states $\eta_1,\eta_2$ where $\mathcal{F}(\eta_1,\eta_2) = \big( \text{Tr}\sqrt{\sqrt{\eta_1}\eta_2\sqrt{\eta_1}}  ~\big)^2$.
 In general, it is a hard problem. 
 However, when $\rho_{_N}:= \ketbra{\psi_N}{\psi_N}$ is pure, $\mathcal{F}(\rho_{_N},\sigma) = \bra{\psi_{N}}\sigma\ket{\psi_N}$ becomes a linear functional of $\sigma$. This right away implies that $\mathcal{F}(\rho_{_N},\sigma)$ is optimized for extreme points of the set $\mathcal{G}_{N-\text{sep}}$, i.e., pure product Gaussian states of $N$-modes.
 Therefore, for pure state $\rho_{_N}$, Eq. \eqref{eq:tme-gen} simplifies to
 \begin{eqnarray}
     \mathcal{E}_N^{\mathcal{G}}(\rho_{_N})&=& 1 -  \underset{{\sigma \in \mathcal{G}_{N-\text{prod}}}}{\max} \mathcal{F}(\rho_{_N},\sigma).
 \end{eqnarray}
Moreover, since we are restricting to states with vanishing displacements, we have to search through an even smaller subset $\mathcal{G}_{N-\text{prod}}^{\bm{d}=0} \subset \mathcal{G}_{N-\text{prod}}$ of product Gaussian states with vanishing displacements $(\bm{d}=0)$. 
Now, we know that the most general single-mode pure Gaussian state is a squeezed coherent state \cite{GRMP}. Since, for us, the relevant states have vanishing displacements, we can further restrict the possible optimizers.
\begin{eqnarray}
    \sigma = \otimes_{k=1}^N \sigma_k, ~~~\text{with} ~~\sigma_k = \ket{r_k,\theta_k}\bra{r_k,\theta_k},
\end{eqnarray}
where $\ket{r_k,\theta_k}$ is a single-mode squeezed vacuum state with squeezing parameter $r_ke^{i\theta_k}$ with $r_k \in \mathbb{R}\geq0, ~\theta \in [0,2\pi)$. The covariance matrix of $\sigma_k$ is given by
\begin{eqnarray}
    W(r_k,\theta_k) = \begin{bmatrix}
            \cosh (2r_k)-\sinh (2r_k)\cos (\theta_k )&-\sinh (2r_k)\sin (\theta_k )\\ -\sinh (2r_k)\sin (\theta_k )&\cosh (2r_k)+\sinh (2r_k)\cos (\theta_k )
        \end{bmatrix}.\nonumber\\
\end{eqnarray}
The covariance matrix of $\sigma = \otimes_{k=1}^N \sigma_k$ is then computed as $W(\{\vec{r},\vec{\theta}\}) = \bigoplus_{k=1}^N W(r_k,\theta_k).$
On the other hand, we know $\rho_{_N}$ can be completely characterized by its covariance matrix $\Sigma_\rho$. This allows us to express the overlap between $\mathcal{F}(\rho_{_N},\sigma)$ in terms of their covariance matrices \cite{Spedalieri2014}. This enables us to simplify $\mathcal{E}_N^{\mathcal{G}}(\rho_{_N})$ as follows.
\begin{eqnarray}
  \mathcal{E}_N^{\mathcal{G}}(\rho_{_N}) =1 - \max_{\{\vec{r},\vec{\theta}\}}\frac{1}{\sqrt{\det(\Sigma_\rho+  W(\{\vec{r},\vec{\theta}\}))}}.
  \label{eq:purity}
\end{eqnarray}
On the other hand, note that the action of any Gaussian quantum channel on a Gaussian state characterized by displacement vector $\bm{x}$ and covariance matrix $V$ can be described by \cite{EisertWolf2007}
\begin{eqnarray}
    \bm{x} \to X\bm{x}+\bm{h}, ~~ V \to XVX^T+Y,
\end{eqnarray}
where $\bm{h}$ is a fixed real displacement vector, and $X,Y$ are real $2N\times2N$ matrices satisfying
\begin{eqnarray}
    Y + i\Omega - i X \Omega X^{T} \ge 0, ~~\text{where} ~~\Omega = \bigoplus_{k=1}^{N}
\begin{pmatrix}
0 & 1 \\
-1 & 0
\end{pmatrix},
\label{eq:condition}
\end{eqnarray}
is the symplectic form. This condition guarantees complete positivity. In our case, we can identify $\Sigma_\rho \to \Sigma_\rho+W(\{\vec{r},\vec{\theta}\})$ as a Gaussian channel with $\bm{h}=0, ~X = \mathbb{I}, ~Y = W(\{\vec{r},\vec{\theta}\})$. Therefore, $ \Sigma_\rho+W(\{\vec{r},\vec{\theta}\})$ is the covariance matrix of another Gaussian state $\varrho_{_G}(\{\vec{r},\vec{\theta}\})$ parameterized by the squeezing parameters $\{\vec{r},\vec{\theta}\}$.  
Therefore, in terms of symplectic invariants, we have
\begin{eqnarray}
  \mathcal{E}_N^{\mathcal{G}}(\rho_{_N}) =1- \max_{\{\vec{r},\vec{\theta}\}} \frac{1}{\sqrt{\prod_{k=1}^N     \big[\nu_k(\{\vec{r},\vec{\theta}\})\big]^2}} = 1- \max_{\{\vec{r},\vec{\theta}\}} \Big[\Pi_{k=1}^N \nu_k^{-1}(\{\vec{r},\vec{\theta}\}) \Big],
\end{eqnarray}
where $\nu_k(\{\vec{r},\vec{\theta}\})$ are the symplectic eigenvalues of a Gaussian state $\varrho_{_G}(\{\vec{r},\vec{\theta}\})$ with covariance matrix $\Sigma_{\rho}+ W(\{\vec{r},\vec{\theta}\})$. This completes the proof.

\textbf{Remark.} With the identification that $\det(\Sigma_\rho+  W(\{\vec{r},\vec{\theta}\}))^{-1/2}$ is the purity of a Gaussian state $\varrho_{_G}(\{\vec{r},\vec{\theta}\})$ with covariance matrix $\Sigma_\rho+  W(\{\vec{r},\vec{\theta}\}))$, we can rewrite Eq. \eqref{eq:purity} as
\begin{eqnarray}
  \mathcal{E}_N^{\mathcal{G}}(\rho_{_N}) =1 - \max_{\{\vec{r},\vec{\theta}\}} \text{Tr} \varrho^2_{_G}(\{\vec{r},\vec{\theta}\}).
\end{eqnarray}
Therefore, the optimization can be physically thought of as maximizing the purity of a Gaussian state $\varrho_{_G}(\{\vec{r},\vec{\theta}\})$ that is characterized by the squeezing parameters corresponding to $2N$ real parameters.

\section{Generation of random pure Gaussian states}
\label{app:generation_random_state}
In this section, we briefly outline the methodology \cite{Saptarshi25} to generate random pure Gaussian states with bounded energy per mode.
Consider an arbitrary $N$-mode pure Gaussian state whose covariance matrix can be written as
\begin{eqnarray}
\Sigma = O\,\Gamma\,O^{T},
\label{eq:pure_gaussian_CM}
\end{eqnarray}
where $O$ is an orthogonal symplectic matrix, i.e., $O\in K(n):=\mathrm{Sp}(2n,\mathbb{R})\cap O(2n)$, and the group $K(n)$ is isomorphic to the complex unitary group $U(n)$, with $O(2n)$ denoting the real orthogonal group. Interestingly, the Haar measure on ${\tt U}(n)$ in turn induce a Haar measure on ${\tt K}(n)$ where the isomorphism ${\tt U}(n) \to {\tt K}(n)$ is given by
  \begin{eqnarray}
      U \in {\tt U}(n) \to \begin{bmatrix}
          \text{Re}(U) & \text{Im}(U) \\
          -\text{Im}(U) & \text{Re}(U)
      \end{bmatrix} = O(U) \in {\tt K}(n).
      \label{eq:u2o}
  \end{eqnarray}
The orthogonal symplectic matrices preserve the average energy. This follows by noting that Tr $(O \Sigma O^T) =$ Tr $\Sigma$ for all $O \in {\tt K}(n)$. So, the generation of any pure Gaussian state begins with generating an $O$, which we do by Haar uniformly generating $U$ and use Eq. \eqref{eq:u2o} to get $O$. 
This constitutes \textbf{step} $\bm{1}$.
The next step that remains is the generation of $\Gamma$.
The matrix $\Gamma$ corresponds to a tensor product of $N$ single-mode squeezed states and takes the form \(
\Gamma = \mathrm{diag}(x_{_1},x_{_2},\ldots,x_{_N})\oplus
\mathrm{diag}(\frac{1}{x_{1}},\frac{1}{x_{2}},\ldots,\frac{1}{x_{_N}}),
\) with \(x_i\ge 1\).
The energy of each mode is determined by the corresponding squeezing parameter $x_i$. 
Accordingly, we denote the total energy of the system by $E$, given by 
\begin{eqnarray}
    \sum_i^N (x_i+\frac{1}{x_i})=E.
\end{eqnarray}
Now we distribute the energy among the modes as $\{E_i\}_{i=1}^{N}$ such that
\(
\sum_{i=1}^{N} E_i = E, ~\text{with}~ E_i \ge 2.
\)
For a given mode energy $E_i$, the corresponding squeezing parameter $x_i$ is obtained from
\begin{eqnarray}
x_i = \frac{E_i + s_i\sqrt{E_i^{\,2}-4}}{2}, \qquad i=1,\dots,N ,
\label{eq:quad_value_for_each_mode}
\end{eqnarray}
where \(s_i \in \{+1,-1\}\). Whenever $x_i<1$, we replace it by its reciprocal $x_i\rightarrow \frac{1}{x_i}$ so that $x_i\ge1$. This constitutes \textbf{step} $\bm{2}$. Therefore, by combining steps $1$ and $2$ we can generate $O$ and $\Gamma$, thereby generating a random pure Gaussian state by following Eq. \eqref{eq:pure_gaussian_CM}.

%
%
%
\end{document}